\documentclass[a4paper,USenglish,numberwithinsect]{lipics-v2019}

\usepackage{amsmath}
\usepackage{url}
\usepackage{mathtools}

\usepackage{color}
\usepackage{graphicx}
\usepackage{cmll}
\usepackage{stmaryrd}
\usepackage{proof}
\usepackage{amssymb}
\usepackage{bm}

\usepackage{mathtools}
\usepackage{xspace}
\usepackage{cancel}

\newcommand{\locallypostpones}{strongly postpones\xspace}
\newcommand{\localpostponement}{strong postponement\xspace}
\newcommand{\LocalPostponement}{Strong Postponement\xspace}
\newcommand{\Localpostponement}{Strong postponement\xspace}

\newcommand{\deff}{\,:=\,}

\newcommand{\xredx}[2] {\mathrel{{\uset{#1}{\red}}{}_{\mkern-3mu#2}}}

\newcommand{\xbackredx}[2] {\mathrel{{\uset{#1}{\leftarrow}}{}_{\mkern-3mu#2}}}

\newcommand{\revred}{\leftarrow}

\newcommand{\eq}{~=~}

\newcommand{\ww}{\textsf {W}}
\newcommand{\cc}{\textsf {C}}

\newcommand{\hh}{\textsf {H}}

\renewcommand{\ll}{\textsf {L}}

\newcommand{\cbv}{{\mathtt{cbv}}}
\newcommand{\cbn}{{\mathtt{cbn}}}
\newcommand{\betav}{{\beta_v}}

\newcommand{\shufsym}{\mathsf{sh}}
\newcommand{\ex}{\mathsf {e}}
\renewcommand{\int}{\mathsf{i}}

\newcommand{\head}{\mathsf{h}}

\newcommand{\weak}{\mathsf{w}}

\renewcommand{\left}{\mathsf{l}}

\newcommand{\fv}[1]{\mathsf{fv}(#1)}

\newcommand{\sub}[2]{\subs{#1}{#2}}
\newcommand{\var}{x}

\newcommand{\hole}[1]{\langle #1\rangle}

\newcommand{\val}{v}
\newcommand{\valtwo}{w}
\newcommand{\tmtwo}{u}
\newcommand{\tmthree}{s}
\newcommand{\tmfour}{r}
\newcommand{\tmfive}{q}
\newcommand{\tmsix}{p}
\newcommand{\tm}{t}
\newcommand{\tms}{s}
\newcommand{\tmu}{u}
\newcommand{\tmr}{r}
\newcommand{\tmp}{p}
\newcommand{\tmq}{q}

\renewcommand{\to}{\xrightarrow{}}

\newcommand{\tob}{\to_\beta}

\newcommand{\tosh}{\to_\shufsym}

\newcommand{\rootbv}{\mapsto_{\betav}}

\newcommand{\red}{\rightarrow}
\newcommand{\topred}[1]{\mapsto_{#1}}

\newcommand{\topreds}[1]{\topred{\sigma_{#1}}}

\newcommand{\ered}{\uset{\ex}{\red}}
\newcommand{\ired}{\uset{\int}{\red}}

\newcommand{\eredbv} {\mathrel{\eredx{\betav}}}
\newcommand{\nered}{\uset{\neg \ex~}{\red}}
\newcommand{\neredx}[1]{\mathrel{\nered{}_{\mkern-8mu#1}}}

\newcommand{\nereda}{\neredx{\alpha}}
\newcommand{\neredc}{\neredx{\gamma}}

\newcommand{\hred}{\uset{\head}{\red}}
\newcommand{\nhred}{\uset{\neg \head~}{\red}}

\newcommand{\hredx}[1]  {\mathrel{\hred{}_{\mkern-8mu#1}}}
\newcommand{\nhredx}[1]{\mathrel{\nhred{}_{\mkern-8mu#1}}}

\newcommand{\hredb}  {\mathrel{\hredx{\beta}}}
\newcommand{\nhredb}{\mathrel{\nhredx{\beta}}}

\newcommand{\hredo}  {\mathrel{\hredx{\oplus}}}
\newcommand{\nhredo}{\mathrel{\nhredx{\oplus}}}

\newcommand{\nhreda}{\mathrel{\nhredx{\alpha}}}

\newcommand{\hredc}  {\mathrel{\hredx{\gamma}}}
\newcommand{\nhredc}{\mathrel{\nhredx{\gamma}}}

\newcommand{\wred}{\uset{\weak}{\red}}
\newcommand{\nwred}{\uset{\neg \weak~}{\red}}
\newcommand{\wredx}[1]  {\mathrel{\wred{}_{\mkern-8mu#1}}}
\newcommand{\nwredx}[1]{\mathrel{\nwred{}_{\mkern-8mu#1}}}
\newcommand{\wredbv}{\wredx{\betav}}
\newcommand{\nwredbv}{\nwredx{\betav}}

\newcommand{\lred}{\uset{\left}{\red}}
\newcommand{\nlred}{\uset{\neg \left~}{\red}}

\newcommand{\lredx}[1]  {\mathrel{\lred{}_{\mkern-8mu#1}}}
\newcommand{\nlredx}[1]{\mathrel{\nlred{}_{\mkern-8mu#1}}}

\newcommand{\lredbv}  {\mathrel{\lredx{\betav}}}

\newcommand{\lredc}  {\mathrel{\lredx{\gamma}}}
\newcommand{\nlredc}{\mathrel{\nlredx{\gamma}}}

\newcommand{\nlreda}{\mathrel{\nlredx{\alpha}}}

\newcommand{\lreds}[1]  {\mathrel{\lredx{\sigma_{#1}}}}

\newcommand{\sred}{\uset{{\surf}}{\red}}
\newcommand{\nsred}{\uset{{\neg \surf~}}{\red}}

\newcommand{\sredx}[1]  {\mathrel{\sred{}_{\mkern-8mu#1}}}
\newcommand{\nsredx}[1]{\mathrel{\nsred{}_{\mkern-8mu#1}}}

\newcommand{\sredb}  {\mathrel{\sredx{\beta}}}

\newcommand{\sredbv}  {\mathrel{\sredx{\betav}}}

\newcommand{\redbv}{\red_{\beta_v}}

\newcommand{\redb}{\rightarrow_{\beta}}
\newcommand{\eredx}[1]  {\mathrel{\ered{}_{\mkern-8mu#1}}}
\newcommand{\iredx}[1]  {\mathrel{\ired{}_{\mkern-8mu#1}}}

\newcommand{\iredb}{\mathrel{\iredx{\beta}}}

\newcommand{\reda}  {\red_{\alpha}}
\newcommand{\ereda}  {\mathrel{\eredx{\alpha}}}

\newcommand{\ireda}  {\mathrel{\iredx{\alpha}}}

\newcommand{\redc}  {\red_{\gamma}}
\newcommand{\eredc} {\mathrel {\eredx{\gamma}}}
\newcommand{\iredc}  {\mathrel{\iredx{\gamma}}}

\newcommand{\reds}[1]  {\red_{\sigma_{#1}}}
\newcommand{\ereds}[1] {\mathrel {\eredx{\sigma_{#1}}}}

\newcommand{\nereds}[1]  {\mathrel{\neredx{\sigma_{#1}}}}

\newcommand{\eredsh} {\mathrel {\eredx{\shufsym}}}
\newcommand{\neredsh}  {\mathrel{\neredx{\shufsym}}}

\newcommand{\toh}{\uset{\head}{\red}}
\newcommand{\tonh}{\uset{\neg \head}{\red}}

\newcommand{\be}{\beta\eta}

\newcommand{\redeta}  {\red_{\eta}}

\newcommand{\redo}{\rightarrow_{\oplus}}

\newcommand{\parmark}{\circ\mkern -1mu}

\newcommand{\makepar}[1]{~\parmark \mkern-16mu #1}

\newcommand{\iparred}{{\makepar \ired}}

\newcommand{\nhparredb}{{\makepar \tonh}}

\newcommand{\PLambda}{\Lambda_\oplus}

\newcommand{\Val}{\mathcal V}

\renewcommand{\st}{\mid}

\newcommand{\two}{\frac{1}{2}}

\theoremstyle{plain}

\newtheorem{thm}[theorem]{Theorem}
\newtheorem{prop}[theorem]{Proposition}

\newtheorem{property}[theorem]{Property}

\newtheorem*{theorem*}{Theorem}
\newtheorem*{proposition*}{Prop}
\newtheorem*{lemma*}{Lemma}
\newtheorem*{ex*}{Example}

\newtheorem*{cor*}{Cor.}
\newtheorem*{prop*}{Prop}
\newtheorem*{Def*}{Def}

\newcommand{\A}{\mathcal{A}}

\newcommand{\lam}{\lambda}

\newcommand{\ie}{\emph{i.e.}\xspace}
\newcommand{\eg}{\emph{e.g.}\xspace}
\newcommand{\ih}{\emph{i.h.}\xspace}

\usepackage{xcolor}

\newcommand{\set}[1]{\{#1\}}

\newcommand{\iI}{i \in I}

\makeatletter
\newcommand{\uset}[3][0ex]{%
	\mathrel{\mathop{#3}\limits_{
			\vbox to#1{\kern-6\ex@
				\hbox{$\scriptstyle#2$}\vss}}}}
\makeatother

\renewcommand{\l}{\lambda}
\newcommand{\la}[1]{\lambda #1.}

\newcommand{\refl}[1]{Lemma~\ref{l:#1}}
\newcommand{\reflem}[1]{Lemma~\ref{l:#1}}
\newcommand{\reflemma}[1]{Lemma~\ref{l:#1}}

\newcommand{\refthm}[1]{Theorem~\ref{thm:#1}}

\newcommand{\refprop}[1]{Proposition~\ref{prop:#1}}

\newcommand{\refsec}[1]{Sect.~\ref{sec:#1}}

\newcommand{\refapp}[1]{Appendix~\ref{app:#1}}

\newcommand{\reffig}[1]{Fig.~\ref{fig:#1}}

\newcommand{\refex}[1]{Ex.~\ref{ex:#1}}
\newcommand{\reffact}[1]{Property~\ref{fact:#1}}
\newcommand{\defeq}{~:=~}

\newcommand{\vartwo}{y}
\newcommand{\varthree}{z}

\newcommand{\mellies}{{Melli{\`e}s}\xspace}

\newcommand{\rredc}{\mapsto_{\gamma}}
\newcommand{\rreda}{\mapsto_{\alpha}}

\newcommand{\exred}{\uset{\ex}{\red}}
\newcommand{\nexred}{\uset{\neg \ex~}{\red}}

\newcommand{\exredx}[1]  {\mathrel{\exred{}_{\mkern-8mu#1}}}
\newcommand{\nexredx}[1]{\mathrel{\nexred{}_{\mkern-8mu#1}}}

\newcommand{\nexredbv}{\mathrel{\nexredx{\betav}}}
\newcommand{\exredc}{\mathrel{\exredx{\gamma}}}
\newcommand{\nexredc}{\mathrel{\nexredx{\gamma}}}

\newcommand{\subs}[2]{ \{#1{:=}#2\} }

\renewcommand{\AA}{A}

\newcommand{\LP}[2]{\mathtt{SP(#1,#2)}}
\newcommand{\F}[2]{\mathtt{Fact(#1,#2)}}
\newcommand{\PP}[2]{\mathtt{PP(#1,#2)}}
\newcommand{\LS}[2]{\mathtt{lSwap}(#1, #2)}

\renewcommand{\paragraph}{\vspace{-2pt}\subparagraph}

\usepackage{ebproof}

\newcommand{\Z}{Z}

\newcommand{\surf}{\weak}

\renewcommand{\ss}{\ww}
\newcommand{\m}{\mathtt m}

\newcommand{\n}{\mathtt n}

\newcommand{\Red}{\Rightarrow} 
\newcommand{\Redbv}{\Red_{\beta_v} }  
\newcommand{\Redo}{\Red_{\oplus}} 

\newcommand{\MDST}[1]{\mathcal{M}(#1)}

\newcommand{\mdist}[1]{\textbf{[} #1 \textbf{]}}   

\newcommand{\four}{\frac{1}{4}}

\newcommand{\MPLambda}{\MDST{\PLambda}}

\newcommand{\sRed}{\uset{\surf~}{\Red}}

\newcommand{\sRedx}[1]  {\mathrel{\sRed{}_{\mkern-6mu{#1}}}}
\newcommand{\sRedbv}{\sRedx{\betav}}

\newcommand{\iRed}{\uset{\neg\surf~}{\Red}}
\newcommand{\iRedx}[1]  {\mathrel{\iRed{}_{\mkern-6mu#1}}}
\newcommand{\iRedbv}{\iRedx{\betav}}

\title{Factorize Factorization}

\titlerunning{Factorize Factorization} 

\author{Beniamino Accattoli}{Inria \& LIX, \'Ecole Polytechnique, UMR 7161, Palaiseau, France}{}{}{}
\author{Claudia Faggian}{Universit\'e de Paris, IRIF, CNRS, F-75013 Paris, France}{}{}{}
\author{Giulio Guerrieri}{University of Bath, Department of Computer Science, Bath, UK}{}{}{}

\authorrunning{Accattoli, Faggian, Guerrieri} 

\Copyright{Beniamino Accattoli, Claudia Faggian, Giulio Guerrieri} 

\ccsdesc[100]{Theory of computation~Rewrite systems}

\ccsdesc[100]{Theory of computation~Logic}

\keywords{
	Lambda Calculus, Rewriting,
	Reduction Strategies, Factorization
}

\funding{This work is partially supported by
  ANR JCJC grant ``COCA HOLA'' (ANR-16-CE40-004-01),   ANR PRC project PPS ({ANR-19-CE48-0014}), and  EPSRC Project EP/R029121/1 \emph{Typed lambda-calculi with sharing and unsharing}.
}

\relatedversion{} 

\acknowledgements{}

\nolinenumbers 

\EventEditors{Christel Baier and Jean Goubault-Larrecq}
\EventNoEds{2}
\EventLongTitle{29th EACSL Annual Conference on Computer Science Logic (CSL 2021)}
\EventShortTitle{CSL 2021}
\EventAcronym{CSL}
\EventYear{2021}
\EventDate{January 25--28, 2021}
\EventLocation{Ljubljana, Slovenia (Virtual Conference)}
\EventLogo{}
\SeriesVolume{183}
\ArticleNo{22}
\begin{document}
	\maketitle

\begin{abstract}
We    present a new technique for proving   factorization theorems 
 for compound rewriting systems in a modular way, which is inspired by the Hindley-Rosen technique for confluence. 
Specifically, our  technique  is well adapted  to deal with extensions of the call-by-name and call-by-value 
$\lambda$-calculi.

The technique is first developed abstractly. We isolate a  sufficient  condition (called linear swap) for lifting 
factorization  from components to the compound system, and which is compatible with $\beta$-reduction. We then closely 
analyze some common factorization schemas for the $\lambda$-calculus.

Concretely, we apply our technique to diverse extensions of the $\lambda$-calculus, among which de' Liguoro and 
Piperno's non-deterministic $\lambda$-calculus and---for call-by-value---Carraro and Guerrieri's shuffling calculus.  
For both calculi the literature contains factorization theorems. 
 In both cases, we give a new proof which is neat, simpler than the original, and  strikingly shorter.  
\end{abstract}

	\section{Introduction}

The  $\lam$-calculus underlies functional programming languages and, more 
generally, the paradigm of higher-order computation. Through the  years, more  and more   advanced 
 features  have enriched this paradigm, including  control, non-determinism, states,  probabilistic 
 or  quantum features. The well established way to proceed is
   to extend the $\lam$-calculus with  new operators.  Every time, good operational  properties, 
   such as confluence, normalization, or termination, need to be proved.    
   It is evident that the more complex and 
advanced  is the calculus under study, the more the ability to  \emph{modularize the  analyses} of its 
properties is  crucial.

Techniques for modular proofs are  available  for  termination and confluence, 
with a rich literature which 
examines under which conditions these properties lift from modules to the compound system---some
representative papers are 
\cite{Toyama87,Toyama87t,ToyamaKB89,Rusinowitch87,Middeldorp89,Middeldorp89b,Middeldorp90,KuriharaK90,BachmairD86,
Geser90, DoornbosK98, Dershowitz09,
DBLP:conf/tlca/Akama93,DBLP:journals/jfp/Blanqui18},  see Gramlich 
\cite{DBLP:journals/tcs/Gramlich12} for a survey.
Termination and confluence concern the existence and the uniqueness of normal forms, which are the 
results of a computation. When the focus is on \emph{how to compute} the result, that is, on 
identifying reduction strategies with good properties, then only few abstract techniques  are 
currently available (we mention \cite{GonthierLM92,Mellie95,Mellies97}, \cite{Terese}(Ch.8), and 
\cite{Accattoli12})---this paper proposes a new one.

\paragraph{Factorization.}  The most basic property about how to compute is 
\emph{factorization}, whose paradigmatic example is the head 
factorization theorem of the $\lam$-calculus 
(theorem 11.4.6 in Barendregt's book \cite{Barendregt84}):
every $\beta$-reduction sequence $\tm \tob^* \tmtwo$ can be 
re-organized/factorized so as to first reducing head redexes and then everything else---in symbols 
$\tm\toh^*\tonh^*\tmtwo$.  

The study of factorization in 
	$\lam$-calculus goes back to Rosser \cite{Rosser1935}.  Here, we adopt \mellies terminology \cite{Mellies97}; please 
be aware   that  
factorization results are sometimes referred to as  \emph{semi-standardization} (Mitschke in \cite{Mitschke79}), or 
\emph{postponement} (\cite{Terese}), and often simply called 
\emph{standardization}---standardization is however a more sophisticated property (sketched 
below) of which factorization is  a basic instance.

{According to \mellies 
  \cite{Mellies97}, the meaning of factorization is that} the \emph{essential} part
 of a computation can always be separated from its junk. Let's abstract the role of head reduction, 
 by assuming that  computations consists of  steps $\ered$ which are in some sense {\emph{essential}},  and steps 
$\ired$ which are not.  Factorization says that every rewrite  
sequence $\tm \red^* \tms$ can be factorized as $ t \ered^* \tmu \ired^* s$, \ie, as a sequence of essential 
steps 
followed by 
inessential ones.

Well known examples of essential 
reductions are  head  and   leftmost-outermost reduction  for the $\lam$-calculus (see Barendregt \cite{Barendregt84}), or  left  and  weak 
reduction for the call-by-value $\lam$-calculus (see Plotkin \cite{PlotkinCbV} and Paolini and 
Ronchi Della Rocca \cite{parametricBook}). 

Very 
much as confluence, factorization  is a non-trivial property of $\lam$-calculi;  proofs require 
 techniques such as   finite developments \cite{CurryFeys58,Terese},  labeling \cite{Levy78,KlopThesis},  or parallel 
reduction \cite{Takahashi95}.

\paragraph{Uses of Factorization.} Factorization is commonly used as a  \emph{building block} in   
proving more sophisticated properties 
of the \emph{how-to-compute} kind. It is often the main 
ingredient in proofs of \emph{normalization} theorems
\cite{Barendregt84,Takahashi95,HirokawaMM15,AccattoliFG19}, stating that a reduction strategy reaches a normal 
form whenever one exists. Leftmost-outermost normalization is a well known example.

 Another property, \emph{standardization}, generalizes factorization: 
reduction 
sequences can be organized with respect to an order on redexes, not just with respect to the 
distinction essential/inessential. It is an early result that factorization can be used to 
prove standardization: iterated head factorizations provides what is probably the simplest way to 
prove Curry and Feys' left-to-right standardization theorem, via Mitschke's argument  
\cite{Mitschke79}. 

Additionally, the independence of 
some computational tasks, such as garbage collection, is often modeled as a factorization 
theorem.

\paragraph{Contributions of this  Paper.}
In this paper we propose   a technique   for proving  in a 
\emph{modular} way  factorization theorems for \emph{compound {higher-order} systems,}  such as those obtained by extending the 
$\lam$-calculus with advanced features.   The approach  can be seen as an  analogous for factorization of the classical 
technique for confluence based on Hindley-Rosen lemma, which we discuss in the next paragraphs. {Mimicking the use of Hindley-Rosen lemma is natural}, yet  to our knowledge such an approach has never been used before, at least not in the 
$\lam$-calculus literature. Perhaps this is because  a  direct transposition  of  Hindley-Rosen technique  does not work 
with $\beta$ reduction,  as 
we discuss below and in  \refsec{modular}. 

After developing a sharper technique,
   we  apply it  to   {various  known}  extensions of the $\lam$-calculus {which}  do not fit into 
easily 
manageable categories of rewriting systems.  In all our case studies, 
our  novel   proofs  are  neat, {concise}, and  simpler than the originals. 

\paragraph{Confluence via Hindley-Rosen.}\label{sec:HRuses} Let's consider confluence. 
The  simplest modular technique to establish it is based  on  Hindley-Rosen lemma, which 
states that  the \emph{union} of two confluent reductions  $\red_1$ and $\red_2$ is  confluent 
if $\red_1$ and  $\red_2$ satisfy a  commutation property.  This is the technique used in Barendregt's book for proving confluence of 
$\red_{\beta\eta}$ (Thm. 3.3.9 in \cite{Barendregt84}), where it is also stressed that the  proof is  simpler than Curry and Feys' original one.

While the result is basic, Hindley-Rosen technique provides a powerful tool to prove  confluence of  compound systems. 
In 
the literature of $\lam$-calculus, 
we mention for   instance its use in 
the linear-algebraic $\lam$-calculus  \cite{ArrighiD17},   the  probabilistic $\lam$-calculus 
\cite{FaggianRonchi}, the $\Lambda\mu$-calculus \cite{Saurin08}, the shuffling calculus \cite{CarraroG14},  the 
$\lam$-calculus extended 
with lists \cite{Revesz92} or  pattern-matching \cite{BucciarelliKR20}, or with a \texttt{let} construct 
\cite{AriolaFMOW95}.
It is worth to spell-out the  gain.
 Confluence is often a non-trivial property to establish---when higher-order is involved, the proof  of confluence 
requires   sophisticated techniques.  The difficulty \emph{compounds} when extending the $\lam$-calculus with new 
constructs. Still, the problem is often originated by $\beta$ reduction itself, which encapsulates the higher-order 
features of the computation. By using Hindley-Rosen lemma, confluence of $\beta$  is used as a \emph{black box}: one 
\emph{relies} on that---without having to prove it again---to show that the extended calculus is confluent.

\paragraph{Hindley-Rosen and Sufficient Conditions.} There is a subtle distinction between  Hindley-Rosen \emph{lemma}, 
and what we refer to as Hindley-Rosen \emph{technique}.  Hindley-Rosen lemma reduces confluence of a compound system 
to commutation of the components---the modules. To establish commutation, however, is a non-trivial task, because it is 
a {\emph{global} property, that is, it quantifies over \emph{all sequences} of steps}.
The success of the lemma  
in the $\l$-calculus literature stems from the existence of easy to check  
conditions which  suffice to prove  commutation. All the examples which  we have mentioned above indeed satisfy 
Hindley's 
\emph{strong commutation} 
property \cite{HindleyPhD} (Lemma 3.3.6 in \cite{Barendregt84}), where at most one reduction---but not both---may 
require {multiple}
steps to close a  diagram;
commutation  follows by a finitary tiling argument.
Strong commutation  turns Hindley-Rosen lemma into an effective, concrete  proof {technique}.

\paragraph{Modular Factorization, Abstractly.} Here, we present a modular approach to factorization inspired 
by the Hindley-Rosen \emph{technique}. A  formulation  of Hindley-Rosen lemma for factorization is immediate, and is 
indeed  
folklore. But exactly as for confluence, this  reduces factorization of a compound system to 
a property that is difficult to establish, without a real gain.  The crucial point  is finding  suitable conditions that 
can be used in practice.  The \emph{issue} here is that the natural adaptation of strong 
commutation to factorization is---in general---\emph{not} verified by  extensions of the $\l$-calculi, as it does not 
interact well with $\beta$ (see \refex{SPfailure} in \refsec{modular}).
We  identify an alternative   condition---called \emph{linear swap}---which is 
satisfied by  a large variety of interesting examples, turning the approach into an effective, concrete 
 \emph{technique}. Testing the linear swap condition is easy and {combinatorial} in nature, as it is 
a \emph{local} property,  in the sense that only single steps (rather than sequences of steps) need to be manipulated. 
This holds true even when the modules are not confluent, or non-terminating.
The other  key point in our approach  is that  we  \emph{assume} the modules to be factorizing, therefore we can  use 
their 
factorization---that may require non-trivial proof techniques such as parallel reductions or finite 
developments---as a 
black box.

\paragraph{Modular Factorization, Concretely.} We then focus on our target,   how to establish factorization results 
for  extensions of the 
$\l$-calculus.
 Concretely, we start from     $\beta$ reduction, or its call-by-value counterpart $\betav$, and allow  the calculus  
to 
be enriched with  extra rules. 
 Here we  discover a further striking gain: for common factorization schemas such as head 
or weak factorization, verifying the required  linear swap conditions reduces to  checking \emph{a single case}, 
together with the fact that the new rule {behaves  well with respect to substitution}. 
The test for modular factorization which we obtain  is a ready-to-use  and easy recipe that can be applied in a variety 
of cases.

We   illustrate our technique by  providing several examples, chosen  to stress the independence of 
the technique from other rewriting properties. 
In particular, 
we give a \emph{new} and {arguably} \emph{simpler}  proof of two  results from the literature. The first is  head 
factorization for  
the non-deterministic $\l$-calculus by de' Liguoro and Piperno \cite{deLiguoroP95}, that  extends 
the   $\lam$-calculus  with a choice operator  $\oplus$. It is a \emph{non confluent} calculus, and it is 
representative of the class of $\lam$-calculi extended  with a commutative effect, such as  
\emph{probabilistic} choice; indeed, most features and all issues  are already present there, see 
\cite{LagoZ12} for a thorough discussion. 

The second is a new, simplified proof of  factorization for the shuffling calculus---a refinement  of the 
call-by-value $\l$-calculus due to Carraro and Guerrieri \cite{CarraroG14}, whose left factorization is proved by 
Guerrieri, 
Paolini, and Ronchi della Rocca in \cite{GuerrieriPR17}. In this case the $\l$-calculus is extended with extra rules 
but  they are not associated  to a new operator. The resulting calculus is subtle, as it has  critical pairs.

In both cases, the  new  proof is  neat,  conceptually clear, and    strikingly short. The reason why our  proofs are 
only a few lines {long}, whereas the originals require several pages, is exactly that there is no need to "prove again" 
factorization of $\beta$ or $ \betav$.
We just show that $\beta$  (resp. $ \betav$) interacts well with the new rules.

\paragraph{Further Applications: Probabilistic $\lam$-calculi.} 
 The investigation in this paper was   triggered  by concrete needs, namely the study of
strategies for  probabilistic $\lambda$-calculi \cite{FaggianRonchi,Leventis19}. The probabilistic 
structure adds complexity, and  indeed makes the study of factorization painful---
exposing the  need for   tools to make such an   analysis more manageable.  
Our 
technique smoothly applies, providing new concise proofs that are significantly  simpler than the originals---indeed 
surprisingly simple.
These results are however only overviewed in this paper: we sketch the application to the  call-by-value probabilistic 
calculus {by} Faggian and Ronchi della Rocca \cite{FaggianRonchi}, leaving  the technical  details in 
Appendix~\ref{app:proba}. The reason is that, while the application  of our  technique  is simple, the \emph{syntax} of 
probabilistic $\l$-calculi is not---because   
reduction is defined on  (monadic) structures representing probability distributions over terms. 
Aiming at making the paper accessible within the space limits, we 
 prefer  to focus  on examples in a syntax  which is familiar to a wide audience. Indeed, once the technique 
is understood, its application to other settings is immediate, and  in large part automatic. 
  
\paragraph{A Final Remark.}  Like   Hindley-Rosen for confluence, our technique is sufficient but not necessary to 
factorization. 
Still,  its features and wide range of application make it a remarkable tool 
to tame the complexity which is often associated to the analysis of advanced compound  calculi. 
By emphasizing  the  benefits of a modular approach to factorization, we   hope to   prompt  the development of even 
more techniques.

\paragraph{Related work.}
To our knowledge, the only result in the literature about modular techniques for 
factorization is Accattoli's technique for calculi with explicit substitutions 
\cite{Accattoli12}, which relies on termination hypotheses. Our \emph{linear swap} condition 
(page \pageref{eq:LS})
is technically the same as his \emph{diagonal-swap} condition. One  of the insights at the inception of this work is 
exactly that termination in 
\cite{Accattoli12} is used only to establish factorization of each  
single module, but not when combining them. Here we \emph{assume} modules to be factorizing, 
therefore avoiding termination requirements, and obtaining a more widely applicable technique.

Van Oostrom's decreasing diagrams technique \cite{DD} is a powerful and \emph{inherently modular} tool to establish 
confluence and commutation.
Surprisingly, it has not yet been used for factorization, but   steps in this direction have  been presented recently \cite{Oostrom20}.
%

A divide-and-conquer approach is well-studied for termination. The key point is finding 
conditions which  guarantee that the union of terminating relations is terminating.  Several have 
been studied \cite{BachmairD86,Geser90}; the  weakest such condition, namely $\red_2\cdot 
\red_1\subseteq \red_1 \cup \red_2\cdot (\red_1\cup\red_2 )^*$, is introduced by Doornbos and 
von Karger \cite{DoornbosK98}, and then studied by Dershowitz \cite{Dershowitz09}, under the name 
of \emph{lazy commutation}, and by van Oostrom and Zantema \cite{OostromZ12}. Interestingly, lazy 
commutation {is similar to} the  linear swap condition. 

Finally, a somehow orthogonal approach to study extensions of a rewriting system, which 
is  isolating syntactical classes of term rewriting systems that always satisfy a property. While 
confluence is the most studied property (\eg, \cite{Terese}, Ch 10.4),  factorization and 
standardization are also investigated, in particular for left-to-right standardization  (\cite{Terese}, Ch. 8.5.7).

\section{Preliminaries}\label{sec:basics}
In this section we start by recalling some standard definitions and notations in rewriting theory 
(see \eg \cite{Terese} or \cite{Book_AndAllThat}); we provide an overview of  commutation, confluence, and 
factorization. 
Both \emph{confluence} and \emph{factorization} are forms of commutation.

\paragraph{Basics.} 
An \emph{abstract rewriting system (ARS)}  (see  \cite{Terese}, Ch.1)  is a pair $\A=(\AA, \to)$ 
consisting of a set $A$ and a binary relation $\red$ on $A$ whose pairs are written  $t \to s$ and 
called \emph{steps}. 
 We denote   $\red^*$ (resp. $ \red^= $) the  transitive-reflexive  (resp. reflexive) closure of 
$\red$; 
we denote $\leftarrow$ the reverse relation of $\red$, \ie  $\tmtwo \leftarrow \tm$ if  $\tm\red\tmtwo$.  
 If $\red_1,\red_2$ are binary relations on $\AA$ then 
 $\red_1\cdot\red_2$ denotes their composition,
\ie  $\tm \red_1\cdot\red_2 \tms$ iff there exists  $\tmu\in \AA$ such that $\tm \red_1  \tmu \red_2 \tms$.
We write $(A,\{\red_1,\red_2\})$ to denote the ARS $(A,\red)$ where $\red ~=~\red_1 \cup \red_2$.
We freely use  the fact that the transitive-reflexive closure of a relation is a closure operator, 
\ie\ satisfies
\begin{equation}\label{eq:closure}\tag{\textbf{Closure}}
\red \subseteq \red^*,\quad\quad
(\red^*)^* ~=~\red^*, \quad\quad  \red_1 ~\subseteq ~\red_2 \mbox{ implies } \red_1^* ~\subseteq 
~\red_2^*.
\end{equation}
The following property is  an  immediate consequence:
	\begin{equation}\label{eq:closure_U}\tag{\textbf{TR}}
(\red_1 \cup \red_2 )^*~= ~ (\red_1^* \cup \red_2^*)^*.
\end{equation}

\paragraph{Local vs Global Properties.} An important distinction in 
rewriting theory is  between local and global properties. 
A property  of term $\tm$  is \emph{local} if it is quantified over only \emph{one-step reductions} from 
$\tm$; it is \emph{global} if it is 
quantified over all \emph{rewrite sequences} from $\tm$.  Local properties are easier to test, because the analysis  
(usually) involves 
 a finite number of cases.

\paragraph{Commutation.} Two relations $\red_{1}$ and $\red_{2}$  on $\AA$
\emph{commute} if  
${\revred_1}^{*}\cdot {\red_2}^{*}
{~\subseteq~}{\red_2} ^{*}\cdot\, {\revred_1}^{*}$.

\paragraph{Confluence.} 
 A relation $\red$ on $\AA$ is confluent if it commutes with itself. 
 A classic tool to modularize the proof of confluence is Hindley-Rosen lemma.
Confluence of two relations $\red_1$ and $\red_2$ does not 
imply confluence of 
$\red_1\cup \red_2$, however it does if  {they} commute. 
\begin{lemma*}[Hindley-Rosen] Let  $\red_1$ and $\red_2$ be relations on the set $A$.  If $\red_1$ and 
$\red_2$ are confluent and commute with each other, then $\red_1\cup \red_2$ is confluent.
\end{lemma*}

\paragraph{Easy-to-Check  Conditions for Hindley Rosen.} Commutation is a global condition, which is  difficult to 
test. 
What {turns Hindley-Rosen lemma into} an effective, usable \emph{technique}, is the availability of  {local, 
\emph{easy-to-check}} sufficient  conditions. 
One of the simplest but  most useful such conditions is {Hindley's} strong commutation \cite{HindleyPhD}:
\begin{equation}\label{eq:SC}\tag{\textbf{Strong Commutation}}
  \leftarrow_1 \cdot \red_2 ~\subseteq ~   {\red_2}^* \cdot {\leftarrow_1}^=  
\end{equation}

\begin{lemma}[Local test for commutation \cite{HindleyPhD}]\label{l:SC}
 Strong commutation 
  implies commutation.
\end{lemma}
All the extensions of $\lam$-calculus  we cited at page 
\pageref{sec:HRuses} (namely \cite{Barendregt84, ArrighiD17,FaggianRonchi,Saurin08,CarraroG14, Revesz92, 
BucciarelliKR20,AriolaFMOW95}) prove confluence by using Hindley-Rosen lemma via  strong commutation (possibly in its weaker diamond-like  form 
$\leftarrow_1 \cdot \red_2 ~\subseteq ~   {\red_2}^= \cdot {\leftarrow_1}^=  $).

\paragraph{Factorization.}
We now recall definitions and basic facts    on the rewriting property at the center of this paper, factorization.  
	Let $\A=(A,\{\ered,\ired\})$ be an ARS.
\begin{itemize}
	\item 	 The relation   $\red ~= ~\ered 
	\cup  \ired $  satisfies  \textbf{$\ex$-factorization}, written $\F{\ered}{\ired}$, if
\begin{equation}\tag{\textbf{Factorization}}
\F{\ered}{\ired}: \quad (\ered \cup  \ired)^*~ \subseteq ~\ered^* \cdot \ired^*  
\end{equation}

\item The relation $\ired$ \textbf{postpones} after $\ered$,  written $\PP{\ered}{\ired}$, if
\begin{equation}\tag{\textbf{Postponement}}
\PP{\ered}{\ired}: \quad {\ired}^*\cdot {\ered}^*  ~\subseteq~  {\ered}^* \cdot {\ired}^*.	
\end{equation}
\end{itemize}
Postponement  can  be 
formulated in terms of commutation, and viceversa, since clearly ($\ired$ postpones after $\ered$) if and only 
if  ($\xbackredx {\int}{}$ commutes with  $\ered$). Note that reversing $\ired$ introduce an asymmetry between the two 
relations.
It is  an easy result that
$\ex$-factorization is  equivalent to postponement, which is a more convenient way to express it. 
The  following equivalences---which we {shall} use freely---are all well known.
\begin{lemma}		\label{l:factorization_eq}
	For any two relations $\ered,\ired$ the following statements are equivalent:
	\begin{enumerate}
		\item \emph{Semi-local postponement}:   $\ired^*\cdot \ered$  $\subseteq$   $\ered^* \cdot 
\ired^*$ (and its dual $\ired\cdot \ered^*$  $\subseteq$   $\ered^* \cdot \ired^*$).
		\item \emph{Postponement}:  $\PP{\ered}{\ired}$.
		\item \emph{Factorization}: $ \F{\ered}{\ired} $.
	\end{enumerate}
	
\end{lemma}

Another  property {that} we {shall} use freely is the following, which is immediate by the definition of 
postponement and   property \refeq{eq:closure_U} (page \pageref{eq:closure_U}).

\begin{property}\label{l:PP_characterization1} Given a relation  $\iparred$ such that  $\iparred ^*=\ired^*$,   $  \PP 
\ered \ired  $  if and only if $\PP \ered  \iparred $.

\end{property}
A well-known use of the above is to instantiate $\iparred$  with a notion of parallel reduction (\cite{Takahashi95}).


\paragraph{Easy-to-Check Sufficient Condition for Postponement.} Hindley first noted  that a local 
property  implies postponement, hence 
factorization \cite{HindleyPhD}. 
It is immediate to recognize that the property below  is exactly the postponement analog of   strong commutation  in 
\reflemma{SC}; indeed it is the same expression, with 
$\ired \deff \leftarrow_1 $ and $\ered \deff \red_2 $.

We say that $\ired$ \textbf{\locallypostpones} after $\ered$,
 if
\begin{equation}\label{eq:LP}\tag{\textbf{\LocalPostponement}}
 \LP{\ered} {\ired}:\quad \ired \cdot \ered ~\subseteq~\ered^*\cdot  \ired^=
\end{equation}

\begin{lemma}[Local test  for postponement \cite{HindleyPhD}]\label{l:LP} 
	Strong postponement implies postponement: 
\begin{center}
			$\LP{\ered} {\ired}$ implies  $ \PP{\ered} {\ired} $, and so $ \F{\ered} {\ired} $.
\end{center}
\end{lemma}
\Localpostponement is at the heart of several factorization proofs. However
(similarly to the diamond  property 
for confluence) it  can rarely be used \emph{directly}, because  most   interesting 
relations---\eg $\beta$ reduction in $\lam$ calculus---do not  satisfy it.  
Still, its range of application hugely widens by using \reflemma{PP_characterization1}.

 It is instructive  to examine \localpostponement with respect to $\beta$ reduction; this allows us also to  
recall why it   is difficult to establish head factorization for the  $\lam$-calculus.
\begin{example}[$\lam$-calculus  and strong postponement]\label{ex:LP} In view of head factorization, 
 the $\beta$ reduction is  decomposed in head reduction $\hredb$ and its dual $\nhredb$, that is $\redb ~=~ \hredb \cup 
\nhredb$. To prove head factorization is non trivial  precisely  because  
$\LP{\hredb}{\nhredb}$ \emph{does not} hold. 

Consider the following example: 
$ (\la\var \var\var\var) (I \varthree)\nhredb  (\la\var \var\var\var)  \varthree  \hredb  \varthree\varthree\varthree$.
 The   sequence $\tonh\toh$  can only postpone to a reduction sequence of shape $\toh\toh\tonh\tonh$
\begin{center}
$(\la\var   \var\var\var) (I \varthree)\hredb  (I \varthree)(I \varthree)(I \varthree) \hredb  
\varthree(I \varthree)(I \varthree)
\nhredb \varthree\varthree(I \varthree) \nhredb \varthree\varthree\varthree$ 
\end{center}
A solution  
 is to compress sequences of $\tonh$ by introducing an intermediate relation $\nhparredb$ (\emph{internal parallel 
reduction})
such that $\nhparredb^* = \nhredb^*$ and which does  verify \localpostponement.
This is indeed  the core of Takahashi's {technique} \cite{Takahashi95}.   All the work  in \cite{Takahashi95} goes into 
defining   parallel reductions, and proving   $\LP{\hredb}{\nhparredb}$. One indeed has $\nhparredb \,\cdot \hredb \ 
\subseteq \ \hredb\cdot \hredb^* \cdot \nhparredb$.
\end{example}

\section{Modularizing  Factorization}\label{sec:modular}
	All along this section, we assume to have two relations $\reda,\redc$ on the same set $A$, such that  	 
\begin{center}
			 $\reda ~= ~\ereda \cup \ireda$ and $\redc ~= ~\eredc \cup \iredc$.
\end{center}
We define		
	 $\ired \defeq (\ireda \cup\iredc)$ and $\ered \defeq (\ereda    \cup\eredc)$.
	Clearly $\reda \cup \redc \eq \ired \cup \ered$.
	Our goal is obtaining a technique in the style of Hindley-Rosen's for confluence, to establish that  	if   $\reda,\redc$ 
are $\ex$-factorizing then their union also is, that is, $\F{\ered}{\ired}$ holds.

\paragraph{Issues.}
In spite of  the large and fruitful use  in the  $\lam$-calculus literature  of Hindley-Rosen {technique}  to simplify 
the analysis of confluence,   we are not aware of any similar technique in the analysis of factorization. In this 
section we explain why a   transposition of the technique  is not immediate when  $\beta$ reduction is involved.

A direct equivalent of Hindley-Rosen lemma for commutation is folklore. An explicit  proof is   in \cite{DD}. 
Formulated in terms of postponement we {obtain the following statement.}
\begin{lemma}[Hindley-Rosen transposed to factorization] Assume   $\reda $ and $\redc$ are  $\ex$-factorizing relations.	
	Their union 
	$\reda\cup \redc$  satisfies 
	$\F{\ered}{\ired}$ if 
{\small 	\begin{equation}\label{eq:PP}\tag{\textbf\#}
 \PP \eredc \ireda :~~	\ireda^* \cdot \eredc^* \ \subseteq \ \eredc^* \cdot \ireda^*
	\quad	\text{ and } \quad 
 \PP \ereda \iredc :~~		\iredc^* \cdot \ereda^* \ \subseteq \ \ereda^* \cdot \iredc^*
	\end{equation}}
\end{lemma}
Exactly as Hindley-Rosen lemma, the modularization  lemma above  is of no practical use {by itself}, as the pair of 
conditions \eqref{eq:PP}  one has to test are  as global as the original problem.
What we need is to have local  conditions (akin to strong commutation) to turn the lemma into a usable technique.
 One obvious choice  is   \emph{\localpostponement}: 
 {\small\begin{equation}\label{eq:PP2}\tag{\textbf{\#\#}}
 \LP \eredc \ireda:~~ \ireda \cdot \eredc \ \subseteq \ \eredc^* \cdot \ireda^= 
 \quad  \text{ and } \quad 
\LP \ereda \iredc:~~ \iredc \cdot \ereda \ \subseteq \ \ereda^* \cdot \iredc^=.
 \end{equation}}
Clearly,  \refeq{eq:PP2} implies  \refeq{eq:PP} (\reflem{LP}).
 We may hope  to have all the elements for a  postponement analog of Hindley-Rosen  technique, but it is not the case. 
 Unfortunately, conditions \refeq{eq:PP2} usually \emph{do not hold} in extensions of the $\lambda$-calculus.  Let us 
illustrate the issue with an  example, the non-deterministic $\lam$-calculus, that we shall develop formally in Sect. 
\ref{sec:ND}.

 \begin{example}[Issues]\label{ex:SPfailure} Consider the extension of the language of $\lam$-terms with a construct 
$\oplus$ which models non-deterministic choice. The term $\oplus p q$ non-deterministically reduces to either $p$ or 
$q$, \ie   $\oplus p q \redo p$ or  $\oplus p q \redo q$.  The  calculus $(\Lambda, \{\redb,\redo\})$  has two 
reduction 
rules, $\redb$ and $\redo$.  For both, we define head and non-head steps as usual.
 	
 	Consider the following sequence:  $ (\lam x.xxx) (\oplus p q )\nhredo  (\lam x.xxx) p  \hredb ppp.$
  This   sequence $\nhredo \cdot \hredb$  can only postpone to a reduction sequence of shape $\hredb \cdot \hredo \cdot 
\nhredo \cdot \nhredo$:
	\[ (\lam x.xxx) (\oplus p q) \hredb (\oplus p q)(\oplus p q)(\oplus p q) \hredo p (\oplus p q)(\oplus p q) \nhredo p 
p(\oplus p q) \nhredo ppp.\]
Since the $\beta$-step duplicates the redex $\oplus pq$, 
the    condition $\LP \hredb \nhredo: \nhredo \hredb \ \subseteq \ \hredb^* \, \nhredo^=$  \emph{does not hold}. 
The phenomenon is similar to \refex{LP}, but now  moving to parallel reduction is not a solution: the 
problem here is not just compressing steps,  but the fact that by swapping $\nhredo$ and $\hredb$, a \emph{third} relation
 $\hredo$ appears. 
 \end{example}
Note that the problem  above is specific to factorization, and does not appear with confluence.

\paragraph{A  Robust  Condition for Modular Factorization.} Inspired by Accattoli's study of factorization for 
$\l$-calculi with explicit substitutions \cite{Accattoli12}, we consider an alternative sufficient condition for 
modular factorization, which holds in  many examples, as the next sections shall show.
 
  We say that $\ireda$ 
  \textbf{linearly swaps} with $\eredc$ if
  \begin{equation}\label{eq:LS}\tag{\textbf{Linear Swap}}
  \LS\ireda\eredc:~~\ireda \cdot   \eredc  \ \subseteq\  \eredc \cdot \reda^*
  \end{equation}
     Note that, on the right-hand side,  the relation is $\reda^*$, not $\ireda$. 
  This small change will  make a big  difference, and overcome the issue we have seen in \refex{SPfailure} (note that 
there $\iredb  \hredo \ \subseteq \ \hredb \, \redo^*$ holds).
  Perhaps surprisingly, this easy-to-check condition, which is \emph{local} and 
  \emph{linear} in $\eredc$, suffices, and  holds  in a large variety of cases. {Moreover, it holds \emph{directly} 
(even with $ \beta $) that is, without the mediating role of parallel reductions (as it is the instead the case of 
Takahashi's technique, see \refex{LP}).}

  We finally obtain a  modular factorization technique, via the following easy property.
\begin{lemma}\label{l:basic_com}	
		$\red_a\cdot \red_b ~\subseteq~ \red_b \cdot \red_c^*$ implies   $\red_a^*\cdot \red_b   ~\subseteq ~ \red_b 
\cdot\red_c^*$.
\end{lemma}

  \begin{thm}[Modular factorization]\label{thm:HR_modularity}
	Let  $\reda ~= ~(\ereda \cup \ireda)$ and $\redc ~= ~(\eredc \cup \iredc)$ be  $\ex$-factorizing relations. 
	 The   union
	$\reda\cup \redc$  satisfies $\ex$-factorization 
	$\F{\ered}{\ired}$, for   
	$\ered \ := \ \ereda \cup\eredc$, and $\ired \ := \ \ireda \cup\iredc$, if the following linear swaps hold:
		\begin{equation*}
		\LS\ireda\eredc:~~\ireda \cdot   \eredc  \ \subseteq\  \eredc \cdot \reda^* \quad  \text{ and } \quad  \LS\iredc\ereda:~~
		~~\iredc \cdot   \ereda  \ \subseteq\  \ereda \cdot \redc^*		
	\end{equation*}
\end{thm}
\proof
We prove that 
	the assumptions imply 
		$\LP{ \ered}{\ireda^* \cup \iredc^*}$, hence  $\PP{\ered}{\ireda^* \cup \iredc^*}$ (by  \reflem{LP}). 		
	Therefore  
	$\PP \ered \ired$ by  \reflemma{PP_characterization1} 
{(because  $ (\ireda \cup \iredc )^* ~=~ (\ireda^* \cup \iredc^*)^*$ by  property 
	\refeq{eq:closure_U})}, and so 
	  $\F{\ered}{\ired}$ holds.
	
	 To verify     $\LP{ \ereda \cup \eredc}{\ireda^* \cup \iredc^*}$, we   observe that the following holds:
\begin{center}$\iredx{k }^* \cdot \eredx{j} ~ \subseteq ~( \eredx{j} \cup \eredx{k})^*\cdot \iredx{k}^* \mbox{ for all 
} k,j  
	\in \set{\alpha,\gamma}$\end{center}
	\begin{itemize}
		\item \textbf{Case $j=k$.} This is immediate  by $\ex$-factorization of $\reda$ and $\redc$, and  
by \reflem{factorization_eq}.1.
		
		\item  \textbf{Case $j\not= k$.} $\LS\ireda\eredc$ implies 
		$(\ireda)^* \cdot \eredc$ $\subseteq$  $\eredc \cdot\reda^*$, by   \reflemma{basic_com}. 
		Since $\reda$ $\ex$-factorizes, we obtain $(\ireda)^* \cdot \eredc$ $\subseteq$  $\eredc \cdot \ereda^* 
\cdot\ireda^*$.
		Similarly for $\LS\iredc\ereda$.\qedhere
	\end{itemize}

Note that 
in the  proof of \refthm{HR_modularity},  the assumption that $\reda$ and $\redc$ factorize is crucial. Using that, 
together  with   \reflemma{basic_com}, 
we obtain  $\LP{\ered}{\ireda^*}$, that is, $\ireda^*$ postpones after both $\ex$-steps, (and similarly for $\iredc^*$).
  Note also that   $\LS{\ired}{ \ered}$---taken alone---does not imply $\PP{\ered}{ \ired}$.  
For instance, let's consider again \refex{LP}. It is clear that  $ \LS \nhredb \hredb$ holds and yet it does not imply $\F \hredb\nhredb$. Stronger tools, such as parallel reduction or finite development are needed here---there is no magic.

The next sections apply the modularization result to various  $\lam$-calculus extensions.
\paragraph{Linear Postponement.}
We collect here two easy  properties which shall simplify the proof of factorization in several  of the  case studies (use \reflem{basic_com}).
\begin{lemma}[Linear postponement]\label{l:easy_lp} 
	\hfill
	\begin{enumerate}
		
		\item	($\ired \cdot \ered ~\subseteq~ \ered \cdot \ired^*)   ~\Rightarrow~ \LP{\ered}{\ired^*}   ~\Rightarrow~ 
\F{\ered}{\ired}$.
		
		\item 	($\ired  \cdot  \ered ~ \subseteq ~\ered  \cdot \red^= )  ~\Rightarrow~ \LP{\ered}{\ired} ~\Rightarrow~ 
\F{\ered}{\ired}$.
	\end{enumerate}
\end{lemma}

\paragraph{Factorization vs. Confluence.}
Factorization and confluence are  \emph{independent} properties. In 
\refsec{ND} we  apply our modular  factorization  technique to a non-confluent 
calculus. Conversely, $\beta\eta$, which is confluent, does not verify head nor  leftmost factorization, even though both  $\beta$ and $\eta$---separately---do.


\section{Extensions of the Call-by-Name $\lam$-Calculus:    Head Factorization}\label{sec:studies}\label{sec:lambda}
We shall study  factorization theorems for extensions of both of the call-by-name (shortened to CbN) and of the call-by-value (CbV) 
$\lam$-calculus.
The CbN  $\lam$-calculus---also simply known as  $\lam$-calculus---is the set of $\lam$-terms $\Lambda$, equipped with 
the $\beta$-reduction, while the CbV $\lam$-calculus is the set of $\lam$-terms $\Lambda$, equipped with the 
$\betav$-reduction.

In this section, we first revise the language of the $\lam$-calculus; we  then   consider the case  when the  
calculus is \emph{enriched   with new operators}, such as a non-deterministic choice or a fix-point operator---so, 
together with  $\beta$, we have other reduction rules. 
  We  study  in this setting  \emph{head factorization}, which is 
by far the most important and  common  factorization scheme in $\lam$-calculus. 
We show that  here \refthm{HR_modularity}  further simplifies,  providing  an easy,   \emph{ready-to-use} test for head 
factorization
 of compound systems (\refprop{test_h}).
Indeed,   verifying 
 the two  linear swap conditions of \refthm{HR_modularity} now  reduces to a single,  simple test. Such a 
 simplification only relies on $\beta$ and on the properties  of contextual closure; it  holds  independently from the 
specific form of the extra rule.

\subsection{The (Applied) $\lam$-Calculus} 
Since in the next sections we shall extend the $\l$-calculus with new operators, such as  
a non-deterministic choice $ \oplus $ or a  fix-point $Y$, we allow  in the syntax a 
  set of constants,  meant to represent such  operators. So, for instance, in \refsec{ND} we shall see $\oplus$ as a 
constant.  This 
way   factorization results with respect to $\beta$-reduction  can be seen as holding also in the $\l$-calculus with 
extended syntax---this is  absolutely 
harmless.

Note that despite the fact that the classic Barendregt's book \cite{Barendregt84} defines the $\l$-calculus without 
constants (the calculus is pure), other classic references such as Hindley and Seldin's book \cite{HindleyBook}  or 
Plotkin \cite{PlotkinCbV} do include  constants in the language of terms---thus there is nothing exotic in our approach. Following 
{Hindley and Seldin}, when the set of constants is empty, the 
calculus is called \emph{pure}, otherwise \emph{applied}. 

\paragraph{The Language.}
The following grammars generate $\lam$-terms and contexts.
\begin{center}
$				\begin{array}{rl@{\hspace{1.3cm}} rl}
				\tm, \tmp, \tmq, \tmr  ::=  x  \mid c \mid  \lambda x.\tm \mid \tm\tm   & ( \textbf{terms 
$ \Lambda $}) &
				\cc  ::=  \hole{~}   \mid \tm\cc\mid \cc \tm \mid \lambda x.\cc 
				& (\textbf{contexts})\\		
				\end{array}$
\end{center}
where $x$ ranges over a countable set of \emph{variables},  $c$ over a disjoint  (finite, infinite or empty) set of 
constants. Variables and constants 
are \emph{atoms}, terms of shape 
$\tmp\tmq$ are \emph{applications}, and $\lam x. \tmp$ abstractions.  If the constants are $c_1, ..., c_n$, the set 
of terms is sometimes  noted as $\Lambda_{c_1...c_n}$. 

The plugging $\cc\hole\tm$ of a 
term $\tm$ into a context is the operation replacing the only occurrence of a hole $\hole{~}$ in $\cc$ with 
$\tm$, potentially capturing free variables of $\cc$.

A  reduction step  $\redc$ is defined as the contextual closure of  a \emph{root} relation 
$\rredc$ on $\Lambda$, which is  called a \emph{rule}. Explicitly, if  $ \tmr\rredc \tmr'$, then 
$\tm \redc \tms$ holds if $\tm = \cc\hole\tmr$ and 
$\tms = \cc\hole{\tmr'}$,  for some context $\cc$.
The term $\tmr$ is called a $\gamma$-redex. Given two rules $\mapsto_{\alpha},\rredc$  on 
$\Lambda$, the relation  $\red_{\alpha 
	\gamma}$ is $\reda \cup \redc$, which can equivalently be defined as the contextual closure of 
$\mapsto_\alpha \cup \mapsto_\gamma$.

The  \textbf{(CbN) $\bm{\lam}$-calculus is $(\Lambda, \redb)$}, the set of terms together with  
\textbf{$\bm{\beta}$-reduction}  $\redb$, 
defined as the contextual closure of the  $\beta$-rule:
		$(\lam x. \tmp) \tmq \mapsto_{\beta} \tmp\subs x \tmq$
where $\tmp\subs x \tmq$ denotes  capture-avoiding substitution. We silently work modulo 
$\alpha$-equivalence.

\paragraph{Properties of the Contextual Closure.} 
Here we recall  basic properties about contextual 
closures and substitution, preparing the ground for the simplifications  studied  next.

A relation $\looparrowright$ on terms is \emph{substitutive} if 
\begin{equation}\tag{\textbf{substitutive}}
\tmr \looparrowright \tmr' 
\text{ implies } \tmr \subs x \tmq \looparrowright \tmr'\subs x \tmq.
\end{equation}
An obvious induction  on the shape of terms shows the   following ({see Barendregt \cite{Barendregt84}, p. 54}).
\begin{property}[Substitutive]\label{fact:subs} Let $\redc$ be  the contextual closure of $\rredc$.
\begin{enumerate}
	\item\label{fact:subs-function} If $\rredc $ is substitutive then $\redc$ is 
substitutive: $\tmp\redc \tmp'$ implies $\tmp \subs{x}{\tmq} \redc \tmp' \subs{x}{\tmq}$.
	\item\label{fact:subs-argument} If $\tmq\redc \tmq'$ then $\tm\subs{x}{\tmq} \redc^* 
\tm\subs{x}{\tmq'} $.	
\end{enumerate}
\end{property}
We  recall a basic but key  property of contextual closures. If a step $\redc$ is obtained by closure under 
\emph{non-empty context} of a rule $\rredc$, then it preserves the shape of the term:
\begin{property}[Shape preservation]\label{fact:shape}
	Assume $\tm=\cc\hole{\tmr}\red \cc\hole {\tmr'}=\tm'$ and that  context   $\cc$ is \emph{non-empty}. The term  $\tm'$ 
is an 
application (resp. an abstraction) if and only if $\tm$ is. 
\end{property}
Notice that since the closure under \emph{empty} context of $\rredc$ is always  an \emph{essential} step (whatever head, 
left, or weak), \reffact{shape} implies  that non-essential steps always preserve the shape of terms---we spell this out 
in \reffact{istep} and \ref{fact:istep_CbV} in the Appendix.
Please notice that we shall often write $\rredc$ to indicate the step $\redc$ which is obtained by \emph{empty contextual 
closure}.

\paragraph{Head Reduction.} Head contexts are defined as follows:
$$\hh  ::= \lam x_1 \dots \lam x_k.\hole{~} \tm_1\dots \tm_n \quad \quad (\textbf{head contexts})$$
where  $k\geq0$ and $n\geq 0$. A \emph{non-head context} is a context which is not head.
A \emph{head} step $\hredc$ (resp. \emph{non-head} step $\nhredc$)	is defined as the closure under  
head contexts (resp. non-head contexts) of rule $\mapsto_\gamma$.
Obviously, $\redc \eq \hredc  \cup  \nhredc$. 

Note that   the \emph{empty context} $\hole{~}$ is a \emph{head context}. Therefore 
$\rredc~\subseteq~ \hredc$ holds (a fact which we shall freely use) and  \reffact{shape} always applies to  non-head 
steps.

\subsection{Call-by-Name: Head Factorization, Modularly.}\label{sec:modular_head}
Head factorization is of great importance for the {theory of the} CbN $\lam$-calculus, which is why head factorization for $\redb$ is well studied. 
If we consider a calculus $(\Lambda, \redb\cup \redc)$, where $\redc$ is a new reduction added to $\beta$, 
{our modular technique} (\refthm{HR_modularity}) {states} that the  compound system $\redb\cup \redc$ satisfies  head factorization 
 if $\redc$ does, 
 %
%
%
  and  both 
 $\LS \nhredb \hredc$  and  $\LS \nhredc \hredb$   hold. We show that in the head case our technique simplifies even more, reducing 
to the test in 
 \refprop{test_h}.
 
First, we observe that in this case, \emph{any}  linear swap condition can be tested by considering for the head step  only the root 
relation $\mapsto$, that is, only the closure of $\mapsto$ under \emph{empty} context, which  is a head step by definition.
This is expressed in the following lemma, where we  include also a variant,  that shall be useful 
later on. 
\begin{lemma}[{Lifting} root linear swaps]\label{l:head_swaps}Let $\rreda,\rredc$ be  root relations  on 
$\Lambda$. 
\begin{enumerate}
	\item $\nhreda \cdot \rredc \subseteq  {\hredc} \cdot \reda^* $ implies $\LS \nhreda \hredc$.
\item Similarly, $\nhreda \cdot \rredc \subseteq  {\hredc} \cdot \reda^= $ implies  $\nhreda \cdot \redc \subseteq  {\hredc} \cdot \reda^= $.
\end{enumerate}
\end{lemma}

Second, since we are   studying $ \redb\cup \redc$, one of the linear swaps is $\LS \nhredc \hredb$. We show that, 
whatever is    $\redc$, it  linearly swaps with  $\hredb$ as soon as $\rredc$  is \emph{substitutive}.
\begin{lemma}[Swap with $\hredb$] \label{l:swap_after_b} If  $\rredc$ is   substitutive  then  $\LS \nhredc \hredb$ holds. 
\end{lemma}
The proofs of these two  lemmas are  in  \refapp{head}.

Summing up, since head factorization for $\beta$ is known, we obtain the following test to verify  that {the  compound system $\redb\cup \redc$ satisfies head factorization $\F{\hredb\cup\hredc}{\nhredb\cup \nhredc}$}.
 	\begin{prop}[A test for {modular} head factorization]\label{prop:test_h}
Let 
	$\redb$  be $\beta$-reduction and $\redc$ be the contextual closure of a rule $\rredc $.  	Their union 	$\redb \cup \redc$  
satisfies head factorization if:
	\begin{enumerate}
		\item \emph{Head factorization of $\redc$}: $\F{\hredc}{\nhredc}$.
		\item \emph{Root linear swap}: $\nhredb \cdot \rredc  \ \subseteq \ \hredc \cdot\redb^* $.
		\item \emph{Substitutivity}: $\rredc$ is  substitutive.
	\end{enumerate}

	 	\end{prop}
Note that none of the properties concerns $\redb$ alone, as we already know  that  head factorization of $\redb$  holds.
In \refsec{ND} we shall use our test (\refprop{test_h}) to prove head factorization for the non-deterministic 
$\l$-calculus. The full proof is only a few lines long.

We conclude by observing  that \reflem{head_swaps} 
gives  either a proof that the swap conditions hold, or a counter-example. Let us 
give an example of this {latter} use.

\begin{example}[Finding  counter-examples]\label{ex:be}
The test of \refprop{test_h} can also be used to provide a counter-example to head factorization when 
it fails. Let's instantiate $\redc$ with $\redeta$, that is, the contextual closure of rule $\lam x. 
\tm x\mapsto_{\eta} \tm$ if $\var\notin\fv\tm$.
Now, consider the root linear swap:
	$ \tm := \lam x. (II)(Ix) \nhredb \lam x. (II)x \mapsto_{\eta}II  =: \tms $, 
where $I:= \lam z.z$.
Note that $\tm$ has no $\hredx{\eta}$ step, and so the two steps cannot be swapped. The 
reduction sequence above is a \emph{counter-example to  both head   and leftmost 
factorization} for $\be$. Start with the head (and leftmost) 
redex $II$:	$\lam x. (II)(Ix)\hredx{\be}\lam x. I(Ix)$. From $\lam x. I(Ix)$, 
 there is no way to reach 
$\tms$.
 We recall that $\be$ still satisfies leftmost normalization---the proof is non-trivial  
\cite{KlopThesis,Takahashi95,OostromT16,Ishii18}.
\end{example}

	\section{The Non-Deterministic $\lam$-Calculus $\PLambda$}		\label{sec:ND}
	De' Liguoro and Piperno's non-deterministic $\l$-calculus $\PLambda$ is defined in \cite{deLiguoroP95} by extending  
the $\lam$-calculus with a new operator $\oplus$ whose rule  models  \emph{non-deterministic choice}. Intuitively, 
		$\tm\oplus \tmp$ non-deterministically rewrites to either  $\tm$ or  $\tmp$.		
	 Notably, $\PLambda$ is \emph{not confluent}, hence  it is  a good example of the fact that confluence and 
factorization are independent properties.

We briefly recall $\PLambda$ and its features,   then use our technique to give 
a \emph{novel and neat}  proof  of de' Liguoro and Piperno's \emph{head factorization} result 
(Cor. 2.10 in \cite{deLiguoroP95}).
 
\paragraph{Syntax.} 
We  slightly depart from the presentation in \cite{deLiguoroP95}, as we consider $\oplus$ as a 
constant,  and write $\oplus \tm \tmp$ rather than $\tm \oplus \tmp$, 
working as usual for the $\lam$-calculus with constants (see \eg, 
 \cite{HindleyBook}, or  \cite{Barendregt84}, Sec. 15.3)\footnote{Note that there is no loss with respect to the syntax 
in \cite{deLiguoroP95}, where $\oplus$ always comes with two arguments, because such a constraint defines a  sub-system 
which is closed under reduction.}.
Terms and contexts  are  generated  by:
\begin{align*}
		{ {
				\begin{array}{llr c llr}
				\tm,\tmp,\tmq,\tmr  ::=  x \mid \oplus \mid   \lambda x.\tm \mid \tm\tmp   & ( \textbf{terms}
				~\PLambda)&		\quad\quad&
				\cc  ::=  \hole{~}   \mid \tm\cc\mid \cc \tm \mid \lambda x.\cc 
				& (\textbf{contexts})\\		
				\end{array}
				}
}
\end{align*}
As before,  $\redb$ denotes $\beta$-reduction, while the rewrite step  $\redo$ is the contextual  
closure of  the following non-deterministic rule:
$ 	\oplus \tm \tmp\mapsto_{\oplus} \tm$ and $\oplus \tm \tmp\mapsto_{\oplus} \tmp. $
%
\paragraph{Subtleties.}
The  calculus $(\PLambda,\redb\cup\redo)$ is {non trivial}. Of course, $\red$ is not confluent. Moreover, the  
following  examples  from \cite{deLiguoroP95} show that permuting $\beta$ and $\oplus$ steps is delicate.
\begin{itemize}
	
	\item  \emph{$\redo$   creates  $\beta$-redexes.}  For instance, $((\lam x.x)\oplus y) z \redo (\lam x.x) z \redb z$, 
hence the  $\redo$-step  cannot be postponed after $\redb$.
	
	\item\emph{Choice duplication.} Postponing $\redb$ after $\redo$ is also problematic, because $\beta$-steps may 
multiply choices, introducing new results: flipping a coin and duplicating the result is not equivalent to duplicating 
the coin and then flipping twice. For instance,
	let $\tm=(\lam x.xx) (\oplus \tmp\tmq)$. Duplicating first one may have $\tm\redb (\oplus \tmp\tmq)(\oplus \tmp\tmq) \redo \tmq(\oplus 
\tmp\tmq) \redo \tmq\tmp$ {while flipping first one has} $\tm\redo (\lam x.xx)\tmp\redb ~ pp $ {or $\tm\redo (\lam x.xx)\tmq\redb ~ qq $ but in both cases $\tmq\tmp$ cannot be reached}.  
\end{itemize}
 These examples are   significant  as the same issues impact  any calculus with choice effects.

\paragraph{Head Factorization.}
 The head (resp. non-head)\footnote{Non-head steps are called \emph{internal} ($\ired$) in \cite{deLiguoroP95}.}
  rewrite steps $\hredb$ and $\hredo$ (resp. $\nhredb$ and $\nhredo$)
  are defined as the closure by head (resp. non-head) contexts of rules $\mapsto_\beta$ and $\mapsto_\oplus$, 
respectively.
We also set $\hred ~\defeq~ \hredb \cup \hredo$ and  $\nhred  ~\defeq~ \nhredb \cup \nhredo$.
		
De' Liguoro and Piperno prove that despite the failure of confluence, $\PLambda$ satisfies head factorization.
They prove this result via 	
standardization, following  Klop's technique \cite{KlopThesis}.

\begin{thm}[Head factorization, Cor. 2.10 in \cite{deLiguoroP95}]\label{thm:ND_fact}
$\F \hred \nhred$ holds in the non-deterministic $\l$-calculus $\PLambda$.
\end{thm}

\paragraph{A New Proof, Modularly.}		We give a novel, strikingly  simple proof of   $\F \hred \nhred$, simply  by 
proving  that $\redb$ and $\redo$  satisfy the hypotheses of the test for {modular} head factorization (\refprop{test_h}). 
All the ingredients we need  are given by  the following easy lemma. 
\begin{lemma}[{Root linear swaps}]\label{l:oplus_basic}
	\begin{enumerate}
		\item\label{l:boplus_basic-swap}   
		{		$\tm\nhredb \tmp \mapsto_{\oplus} \tmq$ implies  $\tm \mapsto_{\oplus} \cdot  \red_{\beta} ^= \tmq$.}
		
		\item\label{l:ooplus_basic-swap}  
		{ $\tm\nhredo \tmp \mapsto_{\oplus} \tmq$ implies  $\tm \mapsto_{\oplus} \cdot  \red_{\oplus}^= \tmq$.}
		
	\end{enumerate}
\end{lemma}
\proof
	\begin{enumerate}
	\item 	Let $\tmp=\oplus \tmp_1\tmp_2$ and assume $\oplus\tmp_1\tmp_2 \mapsto_{\oplus} \tmp_i=\tmq$, with  $i\in\set{1,2}$.
		Since  $\tm\nhredb \tmp$, by  
		\reffact{shape} (as spelled-out in  \reffact{istep}), $\tm$ has 
		shape $\oplus \tm_1\tm_2$, with $\oplus \tm_1\tm_2 \nhredb \oplus \tmp_1\tmp_2$. Therefore, 
		either    
		$\tm_1\redb \tmp_1$ or  $\tm_2\redb \tmp_2$,  from which   $\tm=\oplus \tm_1\tm_2  \mapsto_{\oplus}   \tm_i 
		\redb^=  \tmp_i  =\tmq$.
		
	\item 	 The proof is the same as above, just {replace $\beta$ with $\oplus$}.\qedhere	\end{enumerate}

\begin{thm}[Testing  head factorization] We have $\F \hred \nhred$ because we have:
\begin{enumerate}
	\item \emph{Head factorization of $\redo$}: $\F{\hredo}{\nhredo}$.

	\item \emph{Root linear swap}: 
		{$\nhredb \cdot \mapsto_{\oplus}\ \subseteq\ \hredo \cdot  \red_{\beta} ^=$.}		\item \emph{Substitutivity}: {$\mapsto_{\oplus}$ is substitutive.}
\end{enumerate}
\end{thm}
		
\proof We  prove the hypotheses of \refprop{test_h}.
		\begin{enumerate}
				
		\item $\nhredo$ linearly postpones after $\hredo$ because lifting  the swap in \reflem{oplus_basic}.\ref{l:ooplus_basic-swap} via \reflem{head_swaps}.2 (with 
$\alpha=\gamma:=\oplus$) gives  $\tm\nhredo \tmp \hredo \tmq  ~\subseteq ~\tm \hredo \cdot  \redo^= \tmq$. \reflem{easy_lp}.2 gives Factorization. 
			\item This is exactly \reflem{oplus_basic}.\ref{l:boplus_basic-swap}.			\item By definition of substitution $(\oplus \tmp_1\tmp_2)\subs x \tmq  = \oplus (\tmp_1 \subs x 
\tmq)(\tmp_2\subs x \tmq)
			\mapsto_{\oplus} \tmp_i\subs x \tmq$.\qedhere
		\end{enumerate}


	
\section{ Extensions of the CbV $\lam$-Calculus: Left and Weak Factorization}
\label{sec:left}\label{sec:CbV}
Plotkin's call-by-value (CbV) $\l$-calculus \cite{PlotkinCbV} is the restriction of the  $\l$-calculus where 
$\beta$-redexes can be fired only when the argument is a \emph{value}, where values are defined by:
\begin{equation*}
	  		\val ~~ ::=  ~~ x   \mid a \mid \lam x. \tm   \quad \mbox{(\textbf{values $ \Val $}) }
\end{equation*}
The CbV $\l$-calculus   is {given by the pair} $(\Lambda, \redbv)$, where  the  reduction step  $\redbv$ is the 
contextual closure  of the following rule  $\mapsto_{\betav}$: $(\la{x}{\tm})\val \rootbv \tm\subs {x}{\val}$ where $\val$ is a value.

\paragraph{Left and Weak Reduction.} In the literature on the CbV $\lam$-calculus, factorization is considered with 
respect to various essential reductions. Usually, the essential reduction is weak, \ie  it does not act under 
abstractions. There are  three main weak schemes: reducing from left to right, as originally done by Plotkin 
\cite{PlotkinCbV}, from right to left, as done for instance by Leroy's ZINC abstract machine \cite{Leroy-ZINC}, or in an 
unspecified non-deterministic order, used for example in \cite{LagoM08}.

Here we focus on the left(-to-right) and  the (unspecified) weak schemes. Left contexts  $\ll$ and weak contexts $\ww$ 
are respectively  defined by
\begin{equation*}
	\ll ::= \hole{~}  \mid  \ll \tm \mid  \val \ll   \quad\quad\quad  \ww ::= \hole{~}  \mid  \ww \tm \mid   \tm \ww
\end{equation*}
Given a rule $\mapsto_\gamma$, a \emph{left} step $\lredc$ (resp., a \emph{weak} step $\wredx{\gamma}$ ) is   its 
closure by left (resp. weak) context. 
A \emph{non-left} step $\nlredc$ (resp. \emph{non-weak} step $\nwredx{\gamma}$) is a step  obtained as the  closure by a 
context which is not left (resp. not weak).

\paragraph{Left/Weak Factorization, Modularly.}\label{sec:modular_left} For both left and weak reductions, we derive a 
test for modular factorization along the same lines as the test for head factorization (\refprop{test_h}).  
Note that we already know that $(\Lambda, \redbv)$ satisfies   left and weak factorization; the former was proved by 
Plotkin \cite{PlotkinCbV}, the latter is folklore---a proof can be found  in {our previous work} \cite{AccattoliFG19}.
 	\begin{prop}[A  test for modular left/weak factorization]\label{prop:test_lw}
Let 
	$\redbv$  be $\betav$-reduction, $\redc$  the contextual closure of a rule $\rredc $, and $\ex\in \{\left,\weak\}$. 
		Their union 	$\redbv \cup \redc$  satisfies $\ex$-factorization if:
	\begin{enumerate}
		\item \emph{$\ex$-factorization of $\redc$}:~ $\F{\exredc}{\nexredc}$.
		\item \emph{Root linear swap}: 
$ \nexredbv \cdot   \rredc \ \subseteq\  \exredc \cdot\redbv^* $.
		\item \emph{Substitutivity}: $\rredc$ is  substitutive.
	\end{enumerate}
	 	\end{prop}
The  easy proof is in Appendix \ref{app:CbV}.

	\section{The Shuffling Calculus}\label{sec:shuffling}
Plotkin's CbV $\lambda$-calculus is usually considered on closed terms. When dealing with open terms, it is well 
known that a mismatch between the operational and the denotational semantics arises, as first pointed out by Paolini 
and 
Ronchi della Rocca \cite{PaoliniR99,DBLP:conf/ictcs/Paolini01,  parametricBook}.
The literature contains several proposals of extensions of $\betav$ reduction to overcome this issue, see Accattoli and 
Guerrieri for discussions \cite{DBLP:conf/aplas/AccattoliG16}.   One such refinement is Carraro and 
Guerrieri's \emph{shuffling calculus} \cite{CarraroG14}, which extends Plotkin's  $\lambda$-calculus with extra rules  
(without adding new operators). These rules are inspired by linear logic proof nets, and are the CbV analogous of 
Regnier's $\sigma$ rules \cite{regnier94}. 
Left factorization  for the shuffling calculus is  studied by Guerrieri, Paolini, and Ronchi della Rocca in 
\cite{GuerrieriPR17}, by adapting Takahashi's technique \cite{Takahashi95}.

We recall the calculus,  then use our technique to give a new 
 proof of  factorization, both left (as in \cite{GuerrieriPR17}) and weak (new). Remarkably, our proofs are \emph{very 
short}, whereas  the original requires several pages (to  define parallel reductions and prove their properties).

\paragraph{The Syntax.}The \emph{shuffling calculus} is simply Plotkin's calculus extended with 
\emph{$\sigma$-reduction} $\reds{}$, that is, the contextual closure of the root relation $\topreds{} ~= ~\topreds{1} 
\cup \topreds{3}$,
where
\begin{align*}
	(\la{x}{\tm})\tmtwo\tmthree &\topreds{1} (\la{x}\tm\tmthree)\tmtwo \ \text{ if } x \notin \fv{\tmthree} &
	\val((\la{x}{\tm})\tmtwo) &\topreds{3} (\la{x}\val\tm)\tmtwo \ \text{ if } x \notin \fv{\val}
\end{align*}
{We write $\reds{i}\!$ for the contextual closure of $\topreds{i}$ (so $\reds{} \,= \ \reds{1} \!\cup \reds{3}$), and 
$\tosh = \redbv \!\cup \reds{}$.

\paragraph{Subtleties.} From a rewriting perspective, the shuffling calculus is an interesting extension of the 
$\l$-calculus because  its intricate rules  do not fit into easy to manage classes of rewriting systems. 
Orthogonal systems have only simple forms of overlaps of redexes. While the  $\l$-calculus is an orthogonal system,  the 
$\sigma$-rules introduce non-trivial overlaps such as the following ones.   Setting  $I \defeq \la{x}x$ and $\delta 
\defeq \la{x} xx$, the term $\delta I\delta$ is a $\sigma_1$-redex and contains the $\betav$-redex $\delta I$; the term 
$\delta(I\delta)(xI)$ is a $\sigma_1$-redex and contains the $\sigma_3$-redex $\delta(I\delta)$, which contains in turn 
the $\betav$-redex $I\delta$.

\paragraph{Left and Weak  Factorization.} Despite all these  traits, the shuffling calculus has  good properties, such 
as confluence \cite{CarraroG14}, and  left  factorization \cite{GuerrieriPR17}. Moreover, $\reds{}$ is terminating 
\cite{CarraroG14}. The tests developed in the previous section allow us to easily prove both left and 
weak factorization. We check  the hypotheses of  \refprop{test_lw};
 all the ingredients  we need are in \refl{rootsCbV} (the easy details  are in  Appendix \ref{app:shuffling}). Note that 
the \emph{empty context is both a left and a weak} context, hence  $\tm \topreds{i} \tmu$ implies  both $\tm \lreds{i} 
\tmu$ and $\tm \wredx{\sigma_i} \tmu$.
\begin{lemma}[{Root linear swaps}] \label{l:rootsCbV}
	\label{l:linear-swap-sigma}	
	Let $\ex\in \{\left,\weak\}$ and $i \in \{1,3\}$. Then:
		\begin{enumerate}
			\item $\neredx{\betav}\cdot \topreds{i} \ \subseteq\ \topreds{i}  \cdot\red_{\betav}$.
			\item $\neredx{\sigma} \cdot \topreds{i} \ \subseteq\ \topreds{i} \cdot\red_{\sigma}$.
		\end{enumerate}	 
\end{lemma}

\begin{thm}[Testing  left (weak)  factorization]\label{thm:left-test-shuffling}	Let $\ex\in \{\left,\weak\}$.
	$\F \eredsh \neredsh$ holds,  as:
\begin{enumerate}
	\item \emph{Left (resp. weak) factorization of $\reds{}$}:~ $\F{\eredx{\sigma}}{\neredx{\sigma}}$.
	\item \emph{Root linear swap}: 
			$\neredx{\betav}  \cdot\topreds{} \ \subseteq\  \eredx{\sigma} \cdot  
\red_{\betav}$.		
	\item
	 {\emph{Substitutivity}: $\mapsto_{\sigma_i}$ is substitutive, for  $i \in \{1,3\}$.}
\end{enumerate}
\end{thm}

\proof We  prove the hypotheses of \refprop{test_lw}:
	\begin{enumerate}
				\item Left (resp. weak) factorization of $\reds{}$ holds because $\nereds{}$ linearly postpones after 
$\ereds{}$: indeed, by \refl{linear-swap-sigma}{.2} 
				$\neredx{\sigma} \cdot \topreds{} \ \subseteq\ \ereds{} \cdot\red_{\sigma}$ 
				and by contextual closure (\refl{left_swaps} with {$\alpha = \gamma = \sigma$}) we have that $ \nereds{} \cdot 
\ereds{} \subseteq  \ereds{} \cdot  \reds{}^= $. We conclude by 
		 \refl{easy_lp}{.2}.
		\item This is    \refl{linear-swap-sigma}{.1}. 
		\item By definition of substitution. The immediate proof is in Appendix \ref{app:shuffling}.		\qedhere
	\end{enumerate}

\section{Non-Terminating Relations}\label{sec:nonterminating}
In this section we provide  examples of the fact that our technique  does not rest on 
 termination hypotheses. We consider fixpoint operators in both CbN and CbV, which have non-terminating 
reductions. Obviously, when  terms are not restricted by types, the operator  is definable, so the example is  slight 
artificial, but we hope clarifying. 

 There  are also cases where the modules are terminating but the compound system is not;  the technique, surprisingly, 
still works. Accattoli gives various examples based on $\l$-calculi with \textbf{explicit substitutions} in 
\cite{Accattoli12}. 
An insight of this paper---not evident in \cite{Accattoli12}---is that \emph{termination is not needed to lift 
factorization} from the modules to the compound system.

\paragraph{CbN Fixpoint, Head Factorization.}
We first consider  the calculus $\beta Y:=(\Lambda, \redb\cup \red_Y)$ which is defined by Hindley 
\cite{Hindley78}\footnote{Head factorization of $\beta Y$ is easily obtained by a high-level argument, as consequence of 
left-normality, see Terese \cite{Terese}, Ch. 8.5. The point that we want to stress here  is that the validity  of 
\emph{linear swaps} is not limited to \emph{terminating} reduction, and $\beta Y$ provides  a simple, familiar 
example.}.
The reduction $\red_Y$ is  the  contextual closure of the root relation 
$Y\tmp\mapsto_{Y}\tmp (Y\tmp)$. Points 1-3 below are all easily established---details in  \refapp{parY}.

\begin{prop}[Testing head factorization  for $\beta Y$]\label{prop:bY}  $ \redb\cup \red_{Y}$ satisfies head 
factorization:
	\begin{enumerate}
		\item \emph{Head factorization of $\red_Y$}: $\F{\hredx{Y}}{\nhredx{Y}}$.
		
		\item \emph{Root linear swap}: 
		$ \nhredb \cdot \mapsto_{Y} \ \subseteq \  \hredx{Y}\cdot \redb^* $.

		\item \emph{Substitutivity}: $\mapsto_{Y}$ is substitutive.
	\end{enumerate}	
\end{prop}

\paragraph{CbV Fixpoint, Weak Factorization.}
We now consider weak factorization and   a CbV  counterpart of the previous example. 
We follow  Abramsky and McCusker \cite{AbramskyM97}, who study a call-by-value PCF with a fixpoint operator $Z$.
Similarly, we 
extend the CbV $\lam$-calculus with  
their reduction $\red_{\Z}$, which  is the  contextual closure of rule
$ \Z\val\mapsto_{\Z}\lam x. \val (\Z\val)x$ where $\val$ is a value. The    calculus 
$\betav \Z$ is therefore $(\Lambda, \redbv\cup \red_{\Z})$.  Points 1-3 below are all immediate{---details in  \refapp{parY}}.
 
 \begin{prop}[Testing  weak factorization for $\betav {\Z}$]\label{prop:bZ} $ \redb\cup \red_{{\Z}}$ satisfies weak 
factorization:
 	\begin{enumerate}
 		\item \emph{Weak factorization of $\red_{\Z}$}:  {$\F{\wredx{Z}}{\nwredx{Z}}$}.
 				
 		\item \emph{Root linear swap}: 
 	 $ \nwredx{\betav} \cdot  \mapsto_{\Z}   \ \subseteq\   \wredx{\Z}\cdot \redbv^*  $.
	 
	  		\item \emph{Substitutivity}: $\mapsto_{{\Z}}$ is substitutive.
 	\end{enumerate}	
 \end{prop}

\section{Further  Applications: Probabilistic Calculi}\label{sec:proba}
In this paper, we  {present} our technique using   examples which are within  the familiar language  of  
$\lam$-calculus.  However 
the core of the technique---\refthm{HR_modularity}---is independent from a specific syntax. It can be used in calculi 
whose objects are richer than $\lam$-terms.  {The probabilistic} $\lam$-calculus is a prime example. 

A  recent line of research \cite{FaggianRonchi,Leventis19} is developing probabilistic calculi  where evaluation is not limited to a deterministic 
strategy.
Faggian and Ronchi della Rocca \cite{FaggianRonchi} 
define two calculi---$\PLambda^\cbv$ and $\PLambda^\cbn$---which model respectively CbV and CbN probabilistic 
higher-order computation,
while being  \emph{conservative extensions} of the  CbV and CbN $\lam$-calculus.  
For both calculi confluence and  factorization (called \emph{standardization} in \cite{FaggianRonchi}) hold. There is however a deep 
\emph{asymmetry} between the two results.
\emph{Confluence} is neatly proved via Hindley-Rosen technique, by relying on the fact  the $\beta$  (resp. $\betav$) 
reduction is confluent.
The proof of \emph{factorization} is instead    laborious: the authors define  a notion of parallel reduction  for the new calculus, and then adapt Takahashi's technique \cite{Takahashi95}. 
 Leventis work \cite{Leventis19} on  call-by-name probabilistic  $ \lam $-calculus suffers a
 similar problem; he proves factorization 
by relying  on the  finite developments method, but the proof is equally laborious. 

Our technique allows for \emph{a neat, concise  proof of factorization}, which reduces  to only  testing a single  
linear swap, {with no need of parallel reductions or finite developments}. To prove factorization turns out to be in 
fact \emph{easier} than proving confluence. 
The technical details---that is, the definition of the calculus and the proof---are in   Appendix \ref{app:proba}.

	\section{Conclusions and Discussions}\label{sec:conclusion}\label{sec:further}
\paragraph{Summary.}
A well-established  approach to model higher-order computation with advanced  features,  is  starting from 
the call-by-name or call-by-value $\lam$-calculus, and enrich it with  new  constructs. 
We   propose  a sharp technique to establish  factorization of  a compound system  
from factorization of its components. As we point out, the natural transposition of Hindley-Rosen technique for 
confluence does not work here, because the obtained  conditions are---in general---not validated by extensions of the 
$\l$-calculus. The turning point is the identification of an alternative sufficient condition, called linear swap. 
Moreover, on common factorization schemes such as head or weak factorization, our technique reduces to a straightforward 
test. Concretely, we apply  our technique to various examples, 
stressing its independence from common simplifying hypotheses such as confluence, orthogonality, and termination.

\paragraph{Black Box and Elementary Commutations.} A key feature of our 
 technique is to  take factorization of the core relations---the modules---as \emph{black boxes}. The 
focus is then on the analysis of the \emph{interaction} between the modules. 
The benefit  is both practical and conceptual:  
  we disentangle the components---and the issues---under study.
 This is especially appealing when dealing with  extensions of the $\l$-calculus, built on top of $\beta$ or $\betav$ 
reduction, because often most of the difficulties come from the higher-order component, that is, $\beta$ or $\betav$ 
itself---whose factorization is {non-trivial to prove but known to hold}---rather than from the added 
features.

Good illustrations of these points are  our proofs of factorization. We stress that:
\begin{itemize}
	\item the proof of factorization
	of the compound system is independent from the specific technique (finite developments, parallel reduction, etc.) used 
to  prove
	 factorization  of the
	modules. 

\item  to verify good interaction  between the modules, it {often} suffices 
to check  \emph{elementary, local commutations}---the linear swaps. 
\end{itemize}
These features  provide a   neat proof-technique  \emph{supporting the development and 
 the analysis} of complex compound systems. 

\paragraph{Conclusions.}
When one wants to model  new computational features,
the calculus is often not given, 
but it has to be designed, in such a way that it satisfies confluence and factorization.  The process of 
\emph{developing} the calculus and the process of \emph{proving} its  good properties,  go hand in hand. If the latter 
is difficult and prone to errors, the former also is.
The black-box approach makes our   technique efficient and accessible also to working scientists who are not specialists 
in rewriting. {And even}  for the  $\lambda$-calculus expert  who  masters    tools such as  finite developments, labeling or 
parallel reduction, it  still  appears desirable to limit  the amount of difficulties. The more advanced and complex are 
the computational systems we study, the more crucial it is to have  reasoning tools as simple to use as possible.



\bibliography{biblio_rewrite}
	
\newpage
\appendix
\section*{APPENDIX}

\refapp{omitted-proofs} collects omitted  details of proofs. In \refapp{proba} we illustrate  a more 
\emph{advanced example} of application of our technique, namely to the \emph{probabilistic $\lam$-calculus}.

\section{Appendix:  Omitted Proofs}
\label{app:omitted-proofs}

\subsection{Head Factorization (\refsec{lambda})}
\label{app:head}
\paragraph{Consequences of \reffact{shape}.}
Note that the empty context $\hole{}$ is a \emph{head context}. Hence for non-head steps if $\cc\hole{\tmr}\nhred 
\cc\hole{\tmr'}$ necessarily  $\cc\not=\hole{}$, and  \reffact{shape} always applies. The following  \emph{key 
property} 
holds. 
\begin{property}[Shape Preservation]\label{fact:istep}  By  \reffact{shape}, $\nhredc$ preserves the shapes of terms:
	\begin{enumerate}
		\item \label{fact:istep-var} \emph{Atoms}: there is no $\tm$ such that $\tm \nhredc a$, for any variable or constant 
 $a$.
		\item\label{fact:istep-abs} $\tm\nhredc \lam x. \tmu_1$ implies $\tm = \lam x. \tm_1$ and $\tm_1\nhredc \tmu_1$.
		\item\label{fact:istep-app} $\tm\nhredc \tmu_1\tmu_2$ implies  $\tm = \tm_1\tm_2$, with 
		$\tm_1\nhredc 
\tmu_1$ (and $\tm_2=\tmu_2$), or 
$\tm_2\redc \tmu_2$ (and $\tm_1=\tmu_1$).
\item {\emph{Redex}: if $\tm \nhredc \tmu$, and  $\tmu$ is a $\beta$-redex, then $\tm$ is a $\beta$-redex.\\
Similarly, if  $\tm \nhredc \tmu$ and $\tmu$ has shape $K \tm_1...\tm_k$ ($K$ a constant),  $\tm$ has the same shape.}
	\end{enumerate}
\end{property}
{Point 4. follows from points 1. to 3. Note that, in particular,  if $\tmu$ is a $\oplus$-redex,  so is  $\tm$.}

\paragraph{Head Factorization, Modularly.}

\begin{lemma*}[\ref{l:head_swaps}, {Lifting} root linear swaps]
Let $\rreda,\rredc$ be  root relations  on $\Lambda$. 
	\begin{enumerate}
		\item $\nhreda \cdot \rredc \subseteq  {\hredc} \cdot \reda^* $ implies $\LS \nhreda \hredc$.
		\item Similarly, $\nhreda \cdot \rredc \subseteq  {\hredc} \cdot \reda^= $ implies  $\nhreda \cdot \redc \subseteq  {\hredc} \cdot \reda^= $.
	\end{enumerate}
\end{lemma*}

\begin{proof}\textbf{(1).} We prove that $\nhreda \cdot \rredc \subseteq  {\hredc} \cdot \reda^* $ implies $\tm \nhreda \tmu \hredc \tms \subseteq \tm \hredc \cdot \reda^* 
\tms$ ,
	 by induction on  the  head context  $\hh= \lam x_1 \dots \lam x_k.\hole{~} \tm_1\dots \tm_n$ of  the reduction 
$\tmu\hredc \tms$.
	\begin{enumerate}
		\item[A.] $\tmu$ is the $\gamma$-redex (\ie $ \hh=\hole{} $, $k=0$, $n=0$).  The claim holds by assumption.
		
		\item[B.] 	 $\tmu = \lam x.\tmu_1$ (\ie, $ k>0$). Immediate by shape preservation 
(\reffact{istep}.\ref{fact:istep-abs}) and the \ih.
		
		\item[C.] $\tmu=\tmu_1\tmu_2$ (\ie, $k=0$, $n>0$). Then 
		$\tmu_1\hredc \tmu_1' $ and $\tmu_1\tmu_2\hredc  \tmu_1'\tmu_2=\tms$.  
		By shape preservation (\reffact{istep}.\ref{fact:istep-app}), there are two  cases. \\
		Case (i): $\tm:= \tm_1\tmu_2$.	By \ih, 
		$\tm_1\hredc \tm_1' $ and $\tm_1'  \reda^* \tmu_1'$. 
		Hence, $\tm_1\tmu_2\hredc  \tm_1'\tmu_2 \reda^* \tmu_1'\tmu_2 =\tms$. \\
		Case (ii): $\tm:= \tmu_1\tm_2$. Immediate, because  $\tm=\tmu_1\tm_2\hredc \tmu_1'\tm_2\reda \tmu_1'\tmu_2 $.
	\end{enumerate}
\textbf{(2).} The proof is similar. The only  minimal difference  is  case (ii) in point (C).
\end{proof}

\begin{lemma*}[\ref{l:swap_after_b}, Swap with $\hredb$ ] If  $\rredc$ is   substitutive then 
	$\LS \nhredc \hredb$ holds. 
\end{lemma*}

\begin{proof}First, note that by \reffact{subs}, we have that $\redc$ is substitutive. 
	We prove  ($\tm\nhredc \tmu\hredb  \tms$ implies $\tm\hredb  \cdot \redc^* \tms$) by using \reflem{head_swaps}.
	
	Let  $\tmu=(\lam x. \tmu_1)\tmu_2\mapsto_{\beta} \tmu_1\subs{x}{\tmu_2} $. By   \reffact{istep}, either (i)
	$\tm=(\lam x. \tmp )\tmu_2$ and $\tmp\nhredc \tmu_1 $ or (ii)
	$\tm=(\lam x. \tmu_1)\tmq$ and $\tmq\redc \tmu_2 $. 
	Case (i): $(\lam x. \tmp)\tmu_2 \hredb  \tmp\subs{x}{\tmu_2}$. By   \reffact{subs} (point 1)  $\tmp\subs{x}{\tmu_2} \redc 
\tmu_1\subs{x}{\tmu_2}$.
	Case (ii): 
	$(\lam x. \tmu_1)\tmq \hredb  \tmu_1\subs{x}{\tmq}$.  By using \reffact{subs}(point 2),
	$\tmu_1\subs{x}{\tmq}\redc^* \tmu_1\subs{x}{\tmu_2}$.
\end{proof}

\subsection{Call-by-Value $\lam$-Calculus (\refsec{CbV})}\label{app:CbV}
\paragraph{Consequences of \reffact{shape}.}
The empty context $\hole{}$ is  \emph{both a left   and a weak context}. Hence  \reffact{shape} always applies to  
non-left and non-weak steps. Consequently:

\begin{property}[Shape Preservation]\label{fact:istep_CbV} Fixed 
	 $\ex\in \{\left, \weak \}$,  
	 $\neredx{\gamma}$ preserves the shape of terms:
	\begin{enumerate}
		\item \label{l:istep_CbV-var} \emph{Atoms}: there is no $\tm$ such that $\tm \neredc a$, for any variable or 
constant  $a$;
		\item\label{l:istep_CbV-lambda} \emph{Abstraction}:  $\tm\neredc \lam x. \tmu_1$ implies $\tm = \lam x. \tm_1$ and 
$\tm_1\redc \tmu_1$;
		\item\label{l:istep_CbV-app} \emph{Application}: $\tm\neredc \tmu_1\tmu_2$ implies  $\tm = \tm_1\tm_2$, with either 
(i) $\tm_1\neredc \tmu_1$ and $\tm_2=\tmu_2$, or (ii) $\tm_2\redc \tmu_2$ and  $\tm_1=\tmu_1$. 
{Moreover, 	in case (ii): if $\val\tm_2\nlred \val\tmu_2 $ ($\val$ a value), then $ \tm_2\nlred \tmu_2 $;}  if  $\tm_1\tm_2\nwredx{\gamma} \tmu_1\tmu_2$, then $\tm_2\nwredx{\gamma} \tmu_2$ (always).
\item  {\emph{Redex}: if $\tm \neredc \tmu$, and $\tmu$ is a $\beta_v$-redex, then $\tm$ also is (as a consequence of points 1. to 3.)}
	\end{enumerate}
\end{property}

\paragraph{Left and Weak Factorization, Modularly.} To prove \refprop{test_lw} we proceed similarly to   \refsec{modular_head}.
\begin{lemma}[Root linear swaps]\label{l:left_swaps}  Let $\rreda,\rredc$ be root  relations  on $\Lambda$. 
\begin{enumerate}
	\item[i.] 	If		$ \tm \nlreda \tmu \rredc \tms \ \subseteq\ \tm \lredc \cdot \reda^* \tms$  then $\LS \nlreda 
\lredc$.
	\item[ii.] If		$ \tm \nwredx{\alpha} \tmu \rredc \tms \ \subseteq\ \tm \wredx{\gamma} \cdot \reda^* \tms$  then $\LS 
{\nwredx{\alpha}} \lredc$.
\end{enumerate}
Similarly	 $\nereda \cdot \rredc \subseteq  {\eredc} \cdot \reda^= $ implies  $\nereda \cdot \redc \subseteq  {\eredc} \cdot \reda^= $, 
with $\ex\in \{\left,\weak\}$
\end{lemma}

\proof \textbf{(i.)} We prove that 	$ \tm \nlreda \tmu \lredc \tms \ \subseteq\ \tm \lredc \cdot \reda^* \tms$, by induction on the  context $\ll$
 of  $\tmu$.
	\begin{enumerate}
		\item $\ll=\hole{}$, \ie  $\tmu \mapsto_\gamma \tms$.  
		The claim holds by hypothesis.
		
		\item $\ll=\ll \tmu_2$,  \ie $\tmu = \tmu_1\tmu_2 \lredc \tms_1\tmu_2 = \tms$ with $\tmu_1 \lredc \tms_1$. 
		By \reffact{istep_CbV}.\ref{l:istep_CbV-app} $\tm = \tm_1\tm_2$ and 
		\begin{enumerate}
			\item either $\tm_1\nlreda \tmu_1$ (and $\tm_2=\tmu_2$); then, by \ih, $\tm_1 \lredc \cdot \reda^* \tms_1$, so 
$\tm= \tm_1 \tmu_2  \lredc \cdot \reda^* \tms_1\tmu_2 $;
			\item or $\tm_2\reda \tmu_2$ (and $\tm_1 = \tmu_1$), so $\tm= \tmu_1\tm_2 \lredc \tms_1\tm_2 \reda 
\tms_1 \tmu_2 $.
		\end{enumerate}
		
		\item $\ll=\val \ll $, \ie $\tmu = \val\tmu_2 \lredc \val\tms_2 = \tms$ with $\tmu_2 \lredc \tms_2$. Then $\tm = \tm_1\tm_2$ and by \reffact{istep_CbV}.\ref{l:istep_CbV-app}
\begin{enumerate}
		\item either $\tm_1 \nlreda \val$ (and $\tm_2 = \tmu_2$).
		 Since $\val$ is a value,
	by \reffact{istep_CbV}.\ref{l:istep_CbV-var}-\ref{l:istep_CbV-lambda}, $\tm_1$ is also a value and so 
 $\tm=\tm_1\tmu_2 \lredc \tm_1\tms_2 \reda \val \tms_2$.
			\item or  $\tm_2\nlreda \tmu_2$ ($\tm_1=\val$). By \ih, $\tm_2 \lredc \cdot \reda^* \tms_2$, so 
$\val \tm_2  \lredc \cdot \reda^*\val\tms_2 $;		
		\end{enumerate}
	\end{enumerate}
\textbf{(ii.)} The proof of (ii.) is similar, but simpler. Case   $\ww=\hole{}$ is the same. Case  $\ww \tmu_2$ is 
exactely like $\ll \tmu_2$,
and case $\tmu_1 \ww$ is symmetric to case (2).

The proof of the last claim is similar.
\qedhere

\begin{lemma}[Swap with $\eredbv$]\label{swap_after_bv} If  $\rredc$ is  substitutive
	then 	 $\LS {\neredc} {\eredbv}$, for $\ex\in \{\left, \weak\}$.
\end{lemma}

\proof We prove $\LS {\nlredc} {\lredbv}$, the other swap is similar.
	By \reflem{left_swaps}, it is enough to prove that $\tm\nlredc \tmu \mapsto_{\betav}  \tms$ implies $\tm \lredbv  
\cdot \redc^* \tms$. We use \reffact{subs} (substitutivity).
	Let  $\tmu=(\lam x. \tmu_1)\val \mapsto_{\betav} \tmu_1\sub{x}{\val} =\tms$. 
	By  \reffact{istep_CbV} (3. and 4.) we have:
	\begin{enumerate}
			\item either $\tm=(\lam x. \tm_1 )\val$ and $\tm_1 \redc \tmu_1 $; thus, $\tm = (\lam x. \tm_1)\val \lredbv  
		\tm_1\sub{x}{\val} \redc \tmu_1\sub{x}{\val} = \tms$, where the $\redc$ step takes place by 
		\reffact{subs}.\ref{fact:subs-function} since $\redc$ is substitutive.
		\item or  $\tm=(\lam x. \tmu_1)\valtwo$ where $\valtwo \redc \val$ and $\valtwo$ is a value; hence, $\tm = 
(\lam x. \tmu_1)\valtwo \lredbv \tmu_1\sub{x}{\valtwo} \redc^* \tmu_1\sub{x}{\val}$, where the $\redc$ steps take place 
by \reffact{subs}.\ref{fact:subs-argument}.
	\qedhere
	\end{enumerate}

\subsection{The Shuffling Calculus (\refsec{shuffling})}\label{app:shuffling}


\begin{property}[Values are closed under substitution]
	\label{fact:value-sub}
	If $\val$ and $\valtwo$ are values, so is $\val\sub{x}{\valtwo}$.
\end{property}

\begin{lemma*}[\refl{linear-swap-sigma},	Root linear swaps] Let  $\gamma\in \{\betav,\sigma\}$.
\begin{enumerate}
		\item 	If $\tm \nlredx{\gamma} \tmtwo \topreds{i} \tmthree$ then  $\tm \topreds{i}  \cdot  \red_{\gamma} \tmthree$, 
for $i \in \{1,3\}$.
	\item 	If $\tm \nwredx{\gamma} \tmtwo \topreds{i} \tmthree$ then  $\tm \topreds{i}  \cdot  \red_{\gamma} \tmthree$, 
for $i \in \{1,3\}$.
\end{enumerate}
The properties above hold for $\redc $  contextual closure of any rule $\rredc$.
\end{lemma*}

\begin{proof} We prove (1).
	 \reffact{istep_CbV} (iterated) gives that if  $\tm \nlredc \tmtwo$   and $\tmtwo$ is a $\sigma_1$-redex  (resp. a $\sigma_3$-redex), so is $\tm$. We  examine the two 
	 cases for $\tmtwo \topreds{i} \tmthree$.
	\begin{description}
		\item[$\sigma_1$:] By hypothesis, $\tmtwo = (\la{x}\tmfive)\tmsix \tmfour \topreds{1} (\la{x}\tmfive\tmfour)\tmsix = 
\tmthree$ with $x \notin \fv{\tmfour}$.
	Since $\tm \nlredc \tmtwo$, by \reffact{istep_CbV}, we have  $\tm = (\la{x}\tmfive')\tmsix' \tmfour'$ and moreover:
		\begin{itemize}
			\item either $\tmfive' \redc \tmfive$ and $\tmfour' = \tmfour$ and $\tmsix' = \tmsix$,
			\item or $\tmfour' \redc \tmfour$ and $\tmfive' = \tmfive$ and $\tmsix' = \tmsix$,
			\item or $\tmsix' \nlredc \tmsix$ and $\tmfive' = \tmfive$ and $\tmfour' = \tmfour$.
		\end{itemize}
		In any case, $\tm = (\la{x}\tmfive')\tmsix' \tmfour' \lreds{1} (\la{x}\tmfive'\tmfour')\tmsix'  
		\nlredc (\la{x}\tmfive\tmfour)\tmsix = \tmthree$, since $x \notin \fv{\tmfour} \subseteq \fv{\tmfour'}$.

		\item [$\sigma_3$:] By hypothesis, $\tmtwo = \val((\la{x}\tmfour)\tmsix)  \topreds{3} (\la{x}\val\tmfour)\tmsix = 
\tmthree$ with $x \notin \fv{\val}$.
		Since $\tm \nlredc \tmtwo$, by \reffact{istep_CbV}, we have  $\tm = \val'((\la{x}\tmfour')\tmsix')$, and moreover:
		\begin{itemize}
			\item either $\val' \nlredc \val$ and $\tmfour' = \tmfour$ and $\tmsix' = \tmsix$,
			\item or $\tmfour' \redc \tmfour$ and $\val' = \val$ and $\tmsix' = \tmsix$,
			\item or $\tmsix' \nlredc \tmsix$ and $\val' = \val$ and $\tmfour' = \tmfour$.
		\end{itemize}
		In any case, $\tm = \val'((\la{x}\tmfour')\tmsix') \lreds{1} (\la{x}\val'\tmfour')\tmsix'  
		\nlredc (\la{x}\val\tmfour)\tmsix = \tmthree$, as $x \notin \fv{\val} \subseteq \fv{\val'}$.
	\end{description}
Like before, the proof of (2) is similar, and simpler. 	
		\qedhere
\end{proof}

	\begin{lemma}[Substitutivity of $\reds{}$]\label{l:sigma-substitutive}
		If $\tm \topreds{i} \tm'$ then $\tm\sub{x}{\val} \topreds{i} \tm'\sub{x}{\val}$, for $i \in \{1,3\}$.
	\end{lemma}
	
	\begin{proof}
		\begin{description}
			\item[$\sigma_1$:] $\tm = (\la{y}\tmfour)\tmthree\tmtwo \mapsto_{\sigma_1} (\la{x}\tmfour\tmtwo)\tmthree = \tm'$ 
			with $y \notin \fv{\tmtwo}$ and we can suppose without loss of generality that $y \notin \fv{\val} \cup \{x\}$.
			Therefore, $\tm\sub{x}{\val} = (\la{y}\tmfour\sub{x}{\val})\tmthree\sub{x}{\val}\tmtwo\sub{x}{\val} 
			\mapsto_{\sigma_1} (\la{y}\tmfour\sub{x}{\val}\tmtwo\sub{x}{\val})\tmthree\sub{x}{\val} = \tm'\sub{x}{\val}$ since $y 
			\notin (\fv{\tmtwo} \smallsetminus \{x\}) \cup \fv{\val} = \fv{\tmtwo\sub{x}{\val}}$.
			\item[$\sigma_3$:] $\tm = \valtwo((\la{y}\tmtwo)\tmthree) \mapsto_{\sigma_3} (\la{y}\valtwo\tmtwo)\tmthree = \tm'$ 
			with $y \notin \fv{\valtwo}$ and we can suppose without loss of generality that $y \notin \fv{\val} \cup \{x\}$.
			Therefore, $\tm\sub{x}{\val} = \valtwo((\la{y}\tmtwo\sub{x}{\val})\tmthree\sub{x}{\val}) \mapsto_{\sigma_3} 
			(\la{y}\valtwo\sub{x}{\val}\tmtwo\sub{x}{\val})\tmthree\sub{x}{\val} = \tm'\sub{x}{\val}$ as $\valtwo\sub{x}{\val}$ is a 
			value (\reffact{value-sub}) and $y \notin (\fv{\valtwo} \smallsetminus \{x\}) \cup \fv{\val} = 
			\fv{\valtwo\sub{x}{\val}}$.
			\qedhere
		\end{description}
	\end{proof}


\subsection{Non-Terminating Relations (\refsec{nonterminating})}
\label{app:parY}
		
		\begin{proposition*}[\ref{prop:bY}. Testing head factorization  for $\beta Y$] $ \redb\cup \red_{Y}$ 
satisfies head factorization:
			\begin{enumerate}
				\item \emph{Head factorization of $\red_Y$}: $\F{\hredx{Y}}{\nhredx{Y}}$.
				
				\item \emph{Root linear swap}: 
				$ \nhredb \cdot \mapsto_{Y} \ \subseteq \  \hredx{Y}\cdot \redb^* $. 		
				
				\item \emph{Substitutivity}: $\mapsto_{Y}$ is substitutive.
			\end{enumerate}	
		\end{proposition*}
		
		\proof
			We  verify  the hypotheses of \refprop{test_h}:
			\begin{enumerate}
				\item  To verify  that the reduction $\red_Y$  satisfies head factorization is routine. 
				\item Assume $\tm\nhredb  Y\tmp \mapsto_{Y} \tmp (Y\tmp) $. By \reffact{shape}  (as spelled-out in  
\reffact{istep}), 
				if $\tm\nhredb  Y\tmp$ then $\tm=Y\tmq$ and  $\tmq\red_{\beta}\tmp$. Hence 
				$ \tm=Y\tmq\hredx{Y} \tmq (Y\tmq)\redb^* \tmp (Y \tmp)$.
				
				\item Simply $ (Y\tmp) \subs{x}{\tmq} = Y(\tmp \subs{x}{\tmq})\mapsto_{Y}   (\tmp\subs{x}{\tmq})(Y 
(\tmp\subs{x}{\tmq}) ) = (\tmp (Y\tmp)) \subs{x}{\tmq}$.\qedhere
			\end{enumerate}

\begin{proposition*}[\ref{prop:bZ}.  Testing weak factorization for $\betav {\Z}$]  $ \redb\cup \red_{{\Z}}$ 
satisfies weak factorization:
	\begin{enumerate}
		\item \emph{Weak factorization of $\red_{\Z}$}:  {$\F{\wredx{Z}}{\nwredx{Z}}$}.
		
		\item \emph{Root linear swap}: 
		$ \nwredx{\betav} \cdot  \mapsto_{\Z}   \ \subseteq\   \wredx{\Z}\cdot \redbv^*  $.
		
		\item \emph{Substitutivity}: $\mapsto_{{\Z}}$ is substitutive.		
	\end{enumerate}	
\end{proposition*}
\proof
	We  prove the hypotheses of \refprop{test_h}:
	\begin{enumerate}
		\item It is easy to verify that  $\nwredx \Z \cdot \wredx \Z ~\subseteq~ \wredx \Z \cdot \nwredx \Z^*$. Then 
		apply \reflem{easy_lp}{.1}.
		\item  Assume $\tm \nwredx{\betav} \Z\val \mapsto_{\Z} \lam x. \val (\Z\val)x   $. By \reffact{shape} (as 
spelled-out in  \reffact{istep_CbV}), if $\tm\nwredx{\betav} \Z\val $ then
		$\tm=\Z\valtwo$ and  $\valtwo\red_{\betav}\val$. 
		So, 
		$ \tm=\Z\valtwo\wredx{\Z} \lam x. \valtwo (\Z\valtwo) x\redbv^* \lam x.  \val (\Z \val)x$.
		\item 
		{Simply $ (\Z\val) \subs{x}{\tmq} = \Z(\val \subs{x}{\tmq})\mapsto_{\Z}   \la\vartwo\val\subs{x}{\tmq}(\Z 
(\val\subs{x}\tmq)\vartwo = (\la\vartwo\val(\Z \val)\vartwo) \subs{x}{\tmq}$.}
		
		\qedhere
		
\end{enumerate}


\section{Appendix: Factorizing  Factorization in  Probabilistic $\lam$-calculus}\label{app:proba}
Faggian and Ronchi della Rocca 
\cite{FaggianRonchi}  define two calculi---$\PLambda^\cbn$ and $\PLambda^\cbv$---which model respectively CbV and 
CbN probabilistic higher-order computation,
and are conservative extensions of the CbN and CbV  $\lam$-calculi. Here we focus on CbV, which is the most relevant 
paradigm for calculi with effects, but the same approach applies to CbN.

 We first recall the syntax of $\PLambda^\cbv$ (we refer to \cite{FaggianRonchi} for background and details), and then 
give a new proof of weak factorization, using our technique  and obtaining a neat, compact proof of 
factorization, which only requires a few lines.

\paragraph{Terms.}  $\PLambda^\cbv$ is a rewrite system where the objects to be rewritten are not terms, but  monadic 
structures on terms, namely multi-distributions \cite{Avanzini}. Intuitively, a multi-distribution 
    represents a probability distribution on the possible reductions from  a term.
 Terms and contexts are the same as for the non-deterministic $\lam$-calculus, but  here we write the $\oplus$ infix, to 
facilitate reference to \cite{FaggianRonchi}.
\emph{Terms and values}  are generated by the grammars
\begin{align*}
M&::=  x \mid \lambda x.M \mid MM \mid M \oplus M\qquad&\mbox{(\textbf{Terms} $ \PLambda $)}\\
V&:=  x \mid \lambda x.M&\mbox{(\textbf{Values} $ \Val $)}
\end{align*}
where $x$ ranges over a countable set of \emph{variables}.
\emph{Contexts} and \emph{weak  contexts} are given by:
\begin{align*}
	\cc & ::=  \hole{~}    \mid \cc M \mid M\cc \mid \lambda x.\cc \mid \cc \oplus M \mid
	M \oplus \cc && \mbox{(\textbf{Contexts})}\\
	\ss&::=  \hole{~}  \mid \ss M \mid M \ss  && \mbox{(\textbf{Weak Contexts})}
\end{align*}
where $\hole{~}$ denotes the \emph{hole} of the context.  

The intended behaviour of 
$M\oplus N$  is to reduce to either $M$ or $N$,
\emph{with equal probability} $\two$.  This is formalized by means of \emph{multi-distributions}.

\paragraph{Multi-distributions.} A \emph{multi-distribution} $\m=\mdist{p_iM_i\st i\in I}$  is a multiset of pairs of the 
form
$pM$, with $p\in]0,1]$,  $M\in \PLambda$, and  $\sum p_i\leq 1$.
We denote by $\MDST \PLambda$ the set of
	all multi-distributions.  
	The sum of multi-distributions
	is denoted by $+$.  The product
	$q\cdot \m$ of a scalar $q$ and a multi-distribution $\m$ is
	defined pointwise $q\mdist{p_iM_i}_{\iI} := \mdist{(qp_i)M_i}_{\iI} $.

\paragraph{The calculus \pmb{$(\MPLambda, \Red)$}.}

The calculus $\PLambda^\cbv$ is the rewrite system  $(\MPLambda, \Red)$ where    $\MPLambda$ is the set of 
multi-distributions on $\PLambda$ and  the relation $\Red\subseteq \MDST{\PLambda}\times \MDST{\PLambda}$ is defined in  
\reffig{reductions} and \reffig{lifting}.  First, we define   one-step reductions from terms to 
multi-distributions---so for 
example, $M\oplus N \red \mdist{\two M, \two N}$. Then, we lift the definition of reduction
to a binary relation on $\MDST{\PLambda}$,  in the natural way---for instance $\mdist{\two (\lam x.x)z, \two (M\oplus 
N)} \Red $
$ \mdist{\two z, \four  M, \four N}$.
Precisely:
\begin{enumerate}
	\item  The reductions 	
	$\redbv,\redo\subseteq \PLambda \times \MPLambda$ are defined  in 	Fig.~\ref{fig:reductions}. 
	 Observe that  the $\oplus$ rule---probabilistic choice---is  closed only under weak 
	contexts (no reduction  in the body of a function nor in the scope  of an operator $\oplus$). Instead, the $\betav$ rule 
is closed under
	general contexts.   Its restriction to closure  under weak context is denoted $\sredb$. 	
	The relation  $\red$  is the union $\redb \cup \redo$, while weak\footnote{In \cite{FaggianRonchi},  a \emph{weak} 
		reduction (resp. weak context) is called \emph{surface}, and hence noted  $\xredx{\mathsf{s}}{}$.} reduction  $\sred$ is  
	 the union of the weak  reductions $\sredbv \cup \redo$. A $\red$-step which is not weak is noted $\nsred$.

	\item 
	The lifting of a relation
	$\red_r \subseteq  \Lambda_\oplus \times \MDST{\Lambda_\oplus}$ to a reduction on multi-distribution is defined  in Fig.~\ref{fig:lifting}. In particular, 
	$\red, \redbv,\redo, \sred,\nsred$ lift to $\Red,\Redbv,\Redo, \sRed, \iRed$.  
		
\end{enumerate}
The  restriction of  $\redo$ to weak contexts is necessary to  have   confluence, see \cite{FaggianRonchi} for a 
discussion. The fact that  
 reduction $\redbv$ is unrestricted guarantees that  the new calculus is a \emph{conservative 
extension} of CbV $\lam$-calculus.

 	\begin{figure}\centering
	\fbox{
		\begin{minipage}{0.97\textwidth}\centering
		 {\footnotesize 	$ 	{\cc \hole{(\lambda x.M)V }\redbv \mdist{\cc\hole{ M \subs x V} }}{	}
			\qquad
			{\ss\hole{M\oplus N} \red_{\oplus} \mdist{\frac{1}{2}\ss(M), \frac{1}{2}\ss(N)}}{} $		\qquad	
\begin{tabular}{c}
	$ \red\,:=\, \redbv \cup \redo $\\[4pt]
		$ \sred\,:=\, \sredx{\betav} \cup \redo $
\end{tabular}
}
	\end{minipage}}	\caption{Reduction Steps }\label{fig:reductions}
		\fbox{
		\begin{minipage}{0.97\textwidth}\centering
			
	{\footnotesize 	$	\infer{\mdist{M}\Red_r \mdist{M}}{} \quad \quad \quad\quad \quad
			\infer{\mdist{M}\Red_r \m}{M\red_r\m}  \quad  \quad  \quad\quad\quad
			\infer{ \mdist{p_{i}M_{i}\mid i\in I} \Red_r ~+_{\iI}~ {p_i\cdot \m_i}} 
			{(\mdist{M_i} \Red_r \m_i)_{\iI} }$}
			
	\end{minipage}}	  	
	\caption{Lifting $\red_r$ to $\Red_r$}\label{fig:lifting}
\vspace*{-2pt}	
\end{figure}

\paragraph{Factorization, Modularly.}
Faggian and Ronchi della Rocca prove---by defining  suitable notions of parallel reduction and internal parallel 
reduction with respect to  $\Redbv \cup \Redo$---that $\PLambda^\cbv$  satisfies $\F{\sRed}{\iRed}$, that is 
		$\m \Red^*\n$ implies $\m \sRed^*\cdot \iRed^* \n$.

This result---there called \emph{finitary surface  standardization}---is central  in \cite{FaggianRonchi} because it is 
  the base of the asymptotic constructions which are the core of that paper. 
\medskip

We now give a  novel, strikingly short proof of the same result, by using \refthm{HR_modularity}.
It turns out that we only  need to verify the following swap, which is immediate to check.
\begin{lemma}\label{l:post_o} 	
	 $M\nsredx{\betav}  \cdot  \redo  \n$ implies   $M \redo \cdot \Redbv \n$.
\end{lemma}

\begin{thm}[Factorization of $\Red$]  Let $\sRed \,:=\,( \sRedbv \cup \Redo)$ and 
	$\iRed \,:=\, (\iRedbv)$. Then $(\MPLambda, \{\Redbv, \Redo\})$ satisfies $\surf$-factorization 
$\F{\sRed}{\iRed}$.
\end{thm}
\begin{proof}
We verify that the  conditions of \refthm{HR_modularity} hold.  Note that $\oplus$ has no internal steps, 
therefore, it suffices to verify only  two conditions:
\begin{enumerate}
	\item weak factorization of $\Redbv$: $\F{\sRedbv}{\iRedbv}$.
	\item $\iRedbv$ linearly swaps with $\Redo$: $\iRedbv\cdot \Redo ~\subseteq~ \Redo \cdot \Redbv^*$.
\end{enumerate}
The other two conditions of \refthm{HR_modularity} hold vacuously, because $\iRedx{\oplus}\eq\emptyset$   (and 
$\sRedx{\oplus}\eq \Redo$).

	\begin{enumerate}
		\item $\F{\sRedbv}{\iRedbv}$ follows from weak factorization of the CbV $\lam$-calculus $\F{\wredbv}{\nwredbv}$ 
		(see  \refsec{CbV}) because clearly  ($\mdist{1M}\Redbv  \mdist{1N}$ if and only if $M \redbv N$), 
	($\mdist{1M}\sRedbv  \mdist{1N}$ if and only if $M \wredbv N$), and similarly   	($\mdist{1M}\iRedbv  \mdist{1N}$ if 
and only if $M \nwredbv N$).
			\item  \reflem{post_o}  implies 
			$\m\iRedbv \cdot  \Redo  \n
		~\subseteq~	\m \Redo \cdot \Redbv \n$, by  the definition of lifting.\qedhere
	\end{enumerate}
\end{proof}

\end{document}